\documentclass{lime}
\usepackage{shortcuts}
\pdfoutput=1

 \title{Almost-Optimal Sublinear-Time Edit Distance in the Low Distance Regime}
\author[1]{Karl Bringmann}
\author[1]{Alejandro Cassis}
\author[1]{Nick Fischer}
\author[2]{Vasileios Nakos}
\affil[1]{Saarland University and Max Planck Institute for Informatics,\\Saarland Informatics Campus, Saarbrücken, Germany}
\affil[2]{RelationalAI}

\begin{document}
\maketitle

\begin{abstract}
We revisit the task of computing the edit distance in sublinear time. In the $(k,K)$-gap edit distance problem we are given oracle access to two strings of length $n$ and the task is to distinguish whether their edit distance is at most~$k$ or at least $K$. It has been established by Goldenberg, Krauthgamer and Saha~(FOCS~'19), with improvements by Kociumaka and Saha~(FOCS~'20), that the $(k,k^2)$-gap problem can be solved in time $\widetilde\Order(n/k + \poly(k))$. One of the most natural questions in this line of research is whether the $(k,k^2)$-gap is best-possible for the running time $\widetilde\Order(n/k + \poly(k))$.

In this work we answer this question by significantly improving the gap. Specifically, we show that in time $\Order(n/k + \poly(k))$ we can even solve the $(k,k^{1+o(1)})$-gap problem. This is the first algorithm that breaks the $(k,k^2)$-gap in this running time. Our algorithm is almost optimal in the following sense: In the low distance regime ($k \le n^{0.19}$) our running time becomes $O(n/k)$, which matches a known $n/k^{1+o(1)}$ lower bound for the $(k,k^{1+o(1)})$-gap problem up to lower order factors.

Our result also reveals a surprising similarity of Hamming distance and edit distance in the low distance regime: For both, the $(k,k^{1+o(1)})$-gap problem has time complexity $n/k^{1\pm o(1)}$ for small $k$.

In contrast to previous work, which employed a subsampled variant of the Landau-Vishkin algorithm, we instead build upon the algorithm of Andoni, Krauthgamer and Onak~(FOCS~'10) which approximates the edit distance in almost-linear time $O(n^{1+\varepsilon})$ within a polylogarithmic factor. We first simplify their approach and then show how to to effectively prune their computation tree in order to obtain a sublinear-time algorithm in the given time bound. Towards that, we use a variety of structural insights on the (local and global) patterns that can emerge during this process and design appropriate property testers to effectively detect these patterns.
\end{abstract}

\begin{funding}
This work is part of the project TIPEA that has received funding from the European Research Council (ERC) under the European Unions Horizon 2020 research and innovation programme (grant agreement No.~850979).
\end{funding}

\setcounter{page}{0}
\newpage
\section{Introduction}
The \emph{edit distance} (also called \emph{Levenshtein distance})~\cite{Levenshtein66} between two strings~$X$ and~$Y$ is the minimum number of character insertions, deletions and substitutions required to transform~$X$ into~$Y$. It constitutes a fundamental string similarity measure with applications across several disciplines, including computational biology, text processing and information retrieval.

Computational problems involving the edit distance have been studied extensively. A textbook dynamic programming algorithm computes the edit distance of two strings of length $n$ in time $\Order(n^2)$~\cite{Vintsyuk68,WagnerF74}. It is known that beating this quadratic time by a polynomial improvement would violate the Strong Exponential Time Hypothesis~\cite{BackursI15,AbboudBW15,BringmannK15,AbboudHWW16}, one of the cornerstones of fine-grained complexity theory. For faster algorithms, we therefore have to resort to \emph{approximating} the edit distance.\footnote{Throughout the paper, by ``algorithm'' we mean ``randomized algorithm with constant error probability''.} A long line of research (starting even before the hardness result emerged) lead to successively improved approximation algorithms: The first result established that in linear time the edit distance can be $\Order(\sqrt n)$-approximated~\cite{LandauMS98}. The approximation ratio was improved to~$\Order(n^{3/7})$ in~\cite{BarYossefJKK04} and to $n^{1/3+\order(1)}$ in~\cite{BatuES06}. Making use of the Ostrovsky-Rabani technique for embedding edit distance into the~$\ell_1$-metric~\cite{OstrovskyR07}, Andoni and Onak~\cite{AndoniO12} gave a \raisebox{0pt}[0pt][0pt]{$2^{\Order(\sqrt{\log n})}$}-approximation algorithm which runs in time~$n^{1+\order(1)}$. Later, Andoni, Krauthgamer and Onak achieved an algorithm in time~$\Order(n^{1+\varepsilon})$ computing a $(\log n)^{O(1/\varepsilon)}$-approximation~\cite{AndoniKO10}. A breakthrough result by Chakraborty, Das, Goldenberg, Kouck{\`y} and Saks showed that it is even possible to compute a constant-factor approximation in strongly subquadratic time~\cite{ChakrabortyDGKS20}. Subsequent work~\cite{KouckyS20,BrakensiekR20} improved the running time to close-to-linear for the regime of near-linear edit distance. Very recently, Andoni and Nosatzki~\cite{AndoniN20} extended this to the general case, showing that in time $\Order(n^{1+\varepsilon})$ one can compute a $f(1/\varepsilon)$-approximation for some function $f$ depending solely on $\varepsilon$.

While for exact algorithms (almost-)linear-time algorithms are the gold standard, for approximation algorithms it is not clear whether a running time of $\Order(n^{1+\varepsilon})$ is desired, as in fact one might hope for \emph{sublinear-time} algorithms. Indeed, another line of research explored the edit distance in the sublinear setting, where one has random access to the strings~$X$ and~$Y$, and the goal is to compute the edit distance between~$X$ and~$Y$ without even reading the whole input. Formally, in the $(k,K)$-gap edit distance problem the task is to distinguish whether the edit distance is at most $k$ or at least $K$. The running time is analyzed in terms of the string length $n$ and the gap parameters $k$ and $K$. Initiating the study of this problem, Batu et al.~\cite{BatuEKMRRS03} showed how to solve the $(k,\Omega(n))$-gap problem in time $\Order(k^2 / n + \sqrt{k})$, assuming $k \le n^{1-\Omega(1)}$.
The aforementioned algorithm by Andoni and Onak~\cite{AndoniO12} can be viewed in the sublinear setting and solves the $(k,K)$-gap problem in time $n^{2+o(1)} \cdot k / K^2$, assuming $K/k = n^{\Omega(1)}$. In a major contribution, Goldenberg, Krauthgamer and Saha~\cite{GoldenbergKS19} showed how to solve the $(k,K)$-gap problem in time\footnote{We write \raisebox{0pt}[0pt][0pt]{$\widetilde\Order(\cdot)$} to hide polylogarithmic factors, that is, $\widetilde\Order(T) = \bigcup_{c \ge 0} O(T \log^c T)$.} \makebox{$\widetilde\Order(nk/K + k^3)$}, assuming $K/k = 1 + \Omega(1)$. Subsequently, Kociumaka and Saha~\cite{KociumakaS20} improved the running time to \makebox{$\widetilde{O}(nk/K + k^2 + \sqrt{nk^5}/K)$}. They focus their presentation on the $(k,k^2)$-gap problem, where they achieve time $\widetilde{O}(n/k + k^2)$.

How far from optimal are these algorithms? It is well-known that $(k,K)$-gap edit distance requires time $\Omega(n/K)$ (this follows from the same bound for Hamming distance). Batu et al.~\cite{BatuEKMRRS03} additionally proved that $(k,\Omega(n))$-gap edit distance requires time $\Omega(\sqrt{k})$. Together, $(k,K)$-gap edit distance requires time $\Omega(n/K + \sqrt{k})$, but this leaves a big gap to the known algorithms. In particular, in the low distance regime (where, say, $K \le n^{0.01}$) the lower bound is $\Omega(n/K)$ and the upper bound is $\widetilde\Order(nk/K)$. Closing this gap has been raised as an open problem by the authors of the previous sublinear-time algorithms.\footnote{See \url{https://youtu.be/WFxk3JAOC84?t=1104} for the open problems raised in the conference talk on~\cite{GoldenbergKS19}, and similarly \url{https://youtu.be/3jfHHEFNRU4?t=1159} for \cite{KociumakaS20}.} In particular, a natural question is to determine the smallest possible gap which can be distinguished within the time budget of the previous algorithms, say \makebox{$\Order(n/k + \poly(k))$}. Is the currently-best $(k,k^2)$-gap barrier penetrable or can one prove a lower bound?

\paragraph{Our Contribution}
We show that $(k,k^{1+o(1)})$-gap edit distance can be solved in time $O(n/k + k^{4+o(1)})$. In the low distance regime ($k \le n^{0.19}$) this runs in the same time as the previous algorithms for the $(k,k^2)$-gap problem, so we significantly improve the gap~\cite{GoldenbergKS19,KociumakaS20}. 
We stress that the quadratic gap seems to be the limit of the techniques of previous work. Formally, we obtain the following results.

\begin{restatable}[Main Theorem]{theorem}{thmmain} \label{thm:main}
Let $2 \leq B \leq k$ be a parameter. The $(k, \Theta(k \log_B(k) \cdot B))$-gap edit distance problem can be solved in time
\begin{equation*}
     \frac nk \cdot (\log k)^{\Order(\log_B(k))} + \widetilde\Order(k^4 \poly(B)).
\end{equation*}
\end{restatable}

\begin{restatable}[Subpolynomial Gap]{corollary}{coroptimaltime} \label{cor:optimal-time}
In time $\Order(n / k + k^{4+o(1)})$ we can solve $(k, k \cdot 2^{\widetilde\Order(\sqrt{\log k})})$-gap edit distance.
\end{restatable}

\begin{restatable}[Polylogarithmic Gap]{corollary}{corpolyloggap} \label{cor:polylog-gap}
For any $\varepsilon \in (0,1)$, in time $\Order(n / k^{1-\varepsilon} + k^{4+o(1)})$ we can solve $(k, k \cdot (\log k)^{\Order(1/\varepsilon)})$-gap edit distance.
\end{restatable}

Note that one can solve the $(k,k^2)$-gap edit distance problem by running our algorithm from Corollary~\ref{cor:optimal-time} for $\bar k := k^{2-o(1)}$. This runs in time $O(n/k^{2-o(1)} + \poly(k))$, which improves the previously best running time of $\widetilde \Order(n/k + \poly(k))$ for the $(k,k^2)$-gap edit distance problem~\cite{GoldenbergKS19,KociumakaS20} by a factor $k^{1-o(1)}$ in the low distance regime.

\paragraph{Edit Distance versus Hamming Distance}
It is interesting to compare our result against the best-possible sublinear-time algorithms for approximating the \emph{Hamming distance}. 
For Hamming distance, it is well known that the $(k,K)$-gap problem has complexity $\Theta(n/K)$, with matching upper and (unconditional) lower bounds, assuming that $K/k = 1 + \Omega(1)$.
In the large distance regime, Hamming distance and edit distance have been \emph{separated} by the~$\Omega(\sqrt{k})$ lower bound for $(k,\Omega(n))$-gap edit distance~\cite{BatuEKMRRS03}, because $(k,\Omega(n))$-gap Hamming distance can be solved in time $O(1)$.

Our results show a surprising \emph{similarity} of Hamming distance and edit distance: In the low distance regime ($k \le n^{0.19}$) the complexity of the $(k,k^{1+o(1)})$-gap problem is $n/k^{1 \pm o(1)}$ for both Hamming distance and edit distance. Thus, up to $k^{o(1)}$-factors in the gap and running time, their complexity is the same in the low distance regime.\footnote{We remark that this similarity does not yet follow from previous sublinear algorithms, since they solve the $(k,k^2)$-gap edit distance in time \raisebox{0pt}[0pt][0pt]{$\widetilde\Order(n/k+\poly(k))$}~\cite{GoldenbergKS19,KociumakaS20}, while $(k,k^2)$-gap Hamming distance can be solved much faster, namely in time $O(n/k^2)$.}

\paragraph{Our Techniques}
To achieve our result, we depart from the framework of subsampling the Landau-Vishkin algorithm~\cite{LandauV88,LandauMS98}, which has been developed by the state-of-the-art algorithms for sublinear edit distance~\cite{GoldenbergKS19,KociumakaS20}. Instead, we pick up the thread from the \emph{almost-linear-time} algorithm by Andoni, Krauthgamer and Onak~\cite{AndoniKO10}: First, we give (what we believe to be) a more accessible view on that algorithm. Subsequently, we design a sublinear version of it by \emph{pruning} certain branches in its recursion tree and thereby avoid spending time on ``cheap'' subproblems. To this end, we use a variety of structural insights on the (local and global) patterns that can emerge during the algorithm and design property testers to effectively detect these patterns.

\paragraph{Comparison to Goldenberg, Kociumaka, Krauthgamer and Saha~\cite{GoldenbergKKS21}} In a recent arXiv paper, Goldenberg et al.\ studied the complexity of the $(k, k^2)$-gap edit distance problem in terms of \emph{non-adaptive} algorithms~\cite{GoldenbergKKS21}. Their main result is an $\Order(n / k^{3/2})$-time algorithm, and a matching query complexity lower bound. We remark that our results are incomparable: For the $(k, k^2)$-gap problem they obtain a non-adaptive algorithm (which is faster for large~$k$),
whereas we present an adaptive algorithm (which is faster for small~$k$). Moreover, we can improve the gap to $(k, k^{1+\order(1)})$, while still running in sublinear time~$\Order(n/k)$ when $k$ is small. Our techniques also differ substantially: Goldenberg et al.\ build on the work of Andoni and Onak~\cite{AndoniO12} and Batu et al.~\cite{BatuEKMRRS03}, whereas our algorithm borrows from~\cite{AndoniKO10}.

\paragraph{Outline}
The rest of this paper is structured as follows. In \cref{sec:preliminaries} we introduce the necessary preliminaries. In \cref{sec:overview} we give a high-level overview of our algorithm, starting with the necessary ingredients from previous work. We omit the proof details in this section and give the detailed proofs in \cref{sec:main}.

\section{Preliminaries} \label{sec:preliminaries}
For integers $i, j$ we write $\range ij = \set{i, i+1, \dots, j-1}$ and $\rangezero j = \range 0j$. Moreover, we set $\poly(n) = n^{\Order(1)}$, $\polylog(n) = (\log n)^{\Order(1)}$ and  $\widetilde\Order(n) = n \polylog(n)$.

\paragraph{Strings}
We usually denote strings by capital letters $X, Y, Z$. The length of a string $X$ is denoted by $|X|$ and we write $X \circ Y$ to denote the \emph{concatenation} of $X$ and $Y$. For a nonnegative integer $i$, we denote by~$X \access i$ the $i$-th character in $X$ (starting with index zero). For integers $i, j$ we denote by $X \range ij$ the substring of $X$ with indices in $\range ij$. In particular, if the indices are out-of-bounds, then we set $X \range ij = X \range{\max(i, 0)}{\min(j, |X|)}$.

For two strings $X, Y$ with equal length, $X$ is a \emph{rotation} of $Y$ if $X \access i = Y \access{(i + s) \bmod |Y|}$ for some integer shift $s$. We say that $X$ is \emph{primitive} if all of the nontrivial rotations of $X$ are not equal to $X$. For a string $P$, we denote by $P^*$ the infinite-length string obtained by repeating $P$. We say that $X$ is \emph{periodic with period $P$} if $X = P^* \range{0}{|X|}$. We also say that~$X$ is \emph{$p$-periodic} if it is periodic with some period of length at most~$p$.

\paragraph{Hamming and Edit Distance}
For two equal-length strings $X, Y$, we define their \emph{Hamming distance $\HD(X, Y)$} as the number of non-equal characters: $\HD(X, Y) = |\{\, i : X \access i \neq Y \access i \,\}|$. For two strings~$X, Y$ (with possibly different lengths), we define their \emph{edit distance $\ED(X, Y)$} as the smallest number of edit operations necessary to transform $X$ into $Y$; here, an edit operation means \emph{inserting}, \emph{deleting} or \emph{substituting} a character.

We also define an \emph{optimal alignment}, which is a basic object in several of the forthcoming proofs. For two strings $X, Y$, an \emph{alignment} between $X$ and $Y$ is a monotonically non-decreasing function $A : \set{0, \dots, |X|} \to \set{0, \dots, |Y|}$ such that $A(0) = 0$ and $A(|X|) = |Y|$. We say that $A$ is an \emph{optimal alignment} if additionally
\begin{equation*}
    \ED(X, Y) = \sum_{i=0}^{|X|-1} \ED(X \access i, Y \range{A(i)}{A(i+1)}).    
\end{equation*}
This definition is slightly non-standard (compare for instance to the definition in~\cite{Gusfield97}), but more convenient for our purposes. Note that the alignments between $X$ and $Y$ correspond to the paths through the standard edit distance dynamic program. In that correspondence, an optimal alignment corresponds to a minimum-cost path.

\paragraph{Trees}
In the following we will implicitly refer to trees $T$ where each node has an ordered list of children. A node is a \emph{leaf} if it has no children and otherwise an \emph{internal} node. The \emph{depth} of a node $v$ is defined as the number of
ancestors of $v$, and the depth of a tree $T$ is the length of the longest root-leaf path. We refer to the subset of nodes with depth $i$ as the \emph{$i$-th level} in $T$.

\paragraph{Approximations}
We often deal with additive \emph{and} multiplicative approximations in this paper. To simplify the notation, we say that $\widetilde x$ is a $(a, b)$-approximation of $x$ whenever $x/a - b \leq \widetilde x \leq a x + b$.
\section{Overview} \label{sec:overview}
We start by reinterpreting the algorithm of Andoni, Krauthgamer and Onak~\cite{AndoniKO10} as a \emph{framework} consisting of a few fundamental ingredients (\cref{sec:overview:sec:framework}). We then quickly show how to instantiate this framework to recover their original algorithm (\cref{sec:overview:sec:ako-algorithm}), and then develop further refinements to obtain our main result (\cref{sec:overview:sec:alg}). The goal of this section is to outline the main pieces of our algorithm; we defer the formal proofs to \cref{sec:main} and the Appendix.

\subsection{The Andoni-Krauthgamer-Onak Framework}\label{sec:overview:sec:framework}
\paragraph{First Ingredient: Tree Distance}
The first crucial ingredient for the framework is a way to split the computation of the edit distance into smaller, independent subtasks. A natural approach would be to divide the two strings into equally sized blocks, compute the edit distances of the smaller blocks recursively, and combine the results. The difficulty in doing this is that the edit distance might depend on a \emph{global alignment}, which determines how the blocks should align and therefore the subproblems are not independent (e.g.~the optimal alignment of one block might affect the optimal alignment of the next block). However, this can be overcome by computing the edit distances of one block in one string with several shifts of its corresponding block in the other string, and combining the results smartly. This type of \emph{hierarchical decomposition} appeared in previous algorithms for approximating edit distance~\cite{AndoniO12,AndoniKO10,BatuEKMRRS03,OstrovskyR07}. In particular, Andoni, Krauthgamer and Onak~\cite{AndoniKO10} define a string similarity measure called the \emph{tree distance}\footnote{In fact, Andoni et al.\ call the measure the \emph{$\mathcal E$-distance}. However, in a \href{https://slidetodoc.com/polylogarithmic-approximation-for-edit-distance-and-the-asymmetric/}{talk} by Robert Krauthgamer he recoined the name to \emph{tree distance}, and we decided to stick to this more descriptive name.} which gives a good approximation of the edit distance and cleanly splits the computation into independent subproblems.

We will define the tree distance for an underlying tree $T$ which we sometimes refer to as the \emph{partition tree}.

\begin{definition}[Partition Tree] \label{def:partition-tree}
Let $T$ be a tree where each node $v$ is labeled with a non-empty interval $I_v$. We call $T$ a \emph{partition tree} if
\begin{itemize}
\item for the root node $v$ we have $I_v = \rangezero n$, and
\item for any node $v$ with children $v_0, \dots, v_{B-1}$, $I_v$ is the concatenation of $I_{v_0}, \dots, I_{v_{B-1}}$. 
\end{itemize}
\end{definition}

For the original Andoni-Krauthgamer-Onak algorithm, we will use a complete $B$-ary partition tree $T$ with $n$ leaves. In particular, the depth of $T$ is bounded by $\log_B(n)$. There is a unique way to label $T$ with intervals $I_v$: For the~$i$-th leaf (ordered from left to right) we set $I_v = \set{i}$; this choice determines the intervals for all internal nodes. For our algorithm we will later focus on the subtree of $T$ with depth bounded by $\Order(\log_B(k))$ (this is again a partition tree, as can easily be checked).

The purpose of the partition tree is that it determines a decomposition of two length\=/$n$ strings $X$ and $Y$: For a node labeled with interval $I_v = \range ij$ we focus on the substrings~$X \range ij$ and $Y \range ij$. In particular, the leaf labels in the partition tree determine a partition into consecutive substrings. With this definition in hand, we can define the tree distance:

\begin{definition}[Tree Distance] \label{def:tree-distance}
Let $X, Y$ be length-$n$ strings and let $T$ be a partition tree. For any node $v$ in $T$ and any shift $s \in \Int$, we set:
\begin{itemize}
\item If $v$ is a leaf with $I_v = \range ij$, then $\TD_{v, s}(X, Y) = \ED(X \range ij, Y \range{i+s}{j+s})$.
\item If $v$ is an internal node with children $v_0, \dots, v_{B-1}$, then
\begin{equation} \label{eq:tree-distance}
    \TD_{v, s}(X, Y) = \sum_{i \in \rangezero B} \min_{s' \in \Int} \,(\TD_{v_i, s'}(X, Y) + 2 \cdot |s - s'| ).
\end{equation}
\end{itemize}
We write $\TD_T(X, Y) = \TD_{v, 0}(X, Y)$ where $v$ is the root node in~$T$, and we may omit the subscript $T$ when it is clear from the context.
\end{definition}

\begin{figure}[t]
    \caption{Illustrates the tree distance $\TD_{v, s}(X, Y)$ at a node $v$ with interval $I_v = [\,i\,.\,.j\,]$, shift $s$ and children $v_0,\dots,v_3$. The dashed lines denote the shift given by $s$. The bold lines show the ``local'' shifts $s'$ for each of the children (explicitly labeled for $v_0$).
    } \label{fig:tree-distance}
    \centering
    \vspace{1.4ex}
    \includegraphics[width=12cm]{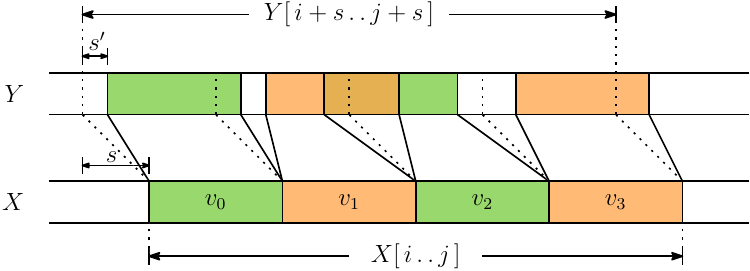}
\end{figure}

\Cref{fig:tree-distance} gives an illustration of this definition. The following lemma (which slightly generalizes the analogous result by~\cite{AndoniKO10}) shows that the tree distance is a useful measure to approximate the edit distance of two strings. We repeat the proof in \cref{sec:ed-td}.

\begin{restatable}[Equivalence of Edit Distance and Tree Distance]{lemma}{lemequivedtd} \label{lem:equivalence-ed-td}
Let~$X, Y$ be strings and let~$T$ be a partition tree with degree at most $B$ and depth at most $D$. Then $\ED(X, Y) \leq \TD_T(X, Y) \leq 2 B D \cdot \ED(X, Y)$.
\end{restatable}

In light of this lemma, we now focus on approximating the tree distance of two strings. The idea behind the Andoni-Krauthgamer-Onak algorithm is to approximately evaluate \cref{def:tree-distance} for all nodes $v$ in the partition tree: For the leaves we directly evaluate the edit distance and for the internal nodes we use \Cref{eq:tree-distance} to combine the recursive computations. However, notice that in \Cref{eq:tree-distance} we minimize over an \emph{infinite} number of shifts $s' \in \Int$. To remedy this situation, recall that we anyways only want to solve a gap problem: Whenever the tree distance exceeds some value $K$ we will immediately report \emph{far}. We therefore restrict our attention to approximating the \emph{capped distances.}

\begin{definition}[Capped Edit Distance] \label{def:capped-ed}
For strings $X, Y$ and $K \geq 0$, we define the \emph{$K$-capped edit distance $\ED^{\leq K}(X, Y) = \min(\ED(X, Y), K)$}.
\end{definition}

\begin{definition}[Capped Tree Distance] \label{def:capped-td}
Let $X, Y$ be strings, let $K \geq 0$ and let $T$ be a partition tree. For any node $v$ in $T$ and any shift $s \in \set{-K, \dots, K}$, we define the \emph{$K$-capped tree distance $\TD^{\leq K}_{v, s}(X, Y)$} as follows:
\begin{itemize}
    \item If $v$ is an a leaf with $I_v = \range ij$, then $\TD_{v, s}^{\leq K}(X, Y) = \ED^{\leq K}(X \range ij, Y \range{i + s}{j + s})$.
    \item If $v$ is an internal node with children $v_0, \dots, v_{B-1}$, then
    \begin{equation} \label{eq:capped-td}
        \TD_{v, s}^{\leq K}(X, Y) = \min\left(\sum_{i \in \rangezero B} \min_{-K \leq s' \leq K} (\TD_{v_i, s'}^{\leq K}(X, Y) + 2 \cdot |s - s'|), K \right)\!.
    \end{equation}
\end{itemize}
We write $\TD_T^{\leq K}(X, Y) = \TD_{v, 0}^{\leq K}(X, Y)$ where $v$ is the root node in $T$, and we omit the subscript $T$ when it is clear from the context.
\end{definition}

It is easy to prove that for computing the $K$-capped tree distance $\TD^{\leq K}(X, Y)$ can be expressed as $\min(\TD(X, Y), K)$, see the following lemma. (The same statement does not apply to $\TD_{v, s}(X, Y)$ for all nodes in the tree.) We provide a formal proof in \cref{sec:ed-td}.

\begin{restatable}[Equivalence of Capped Tree Distances]{lemma}{lemcappedtdequivalence} \label{lem:capped-td-equivalence}
$\TD^{\leq K}(X, Y) = \min(\TD(X, Y), K)$.
\end{restatable}

\paragraph{The Formal Setup}
At this point we are ready to define precisely what we want to compute. We will associate a computational problem to each node in the partition tree $T$. Roughly speaking, the goal for a node $v$ is to approximate the tree distance $\TD_{v, s}^{\leq K}(X, Y)$. We now make the details of this approximation precise. Let us first introduce the following notational shortcuts:
\begin{itemize}
\item $X_v = X \range ij$ (the substring of $X$ relevant at $v$), 
\item $Y_{v, s} = Y \range{i + s}{j + s}$ (the substring of $Y$ relevant at $v$ for one specific shift $s$),
\item $Y_v = Y \range{i - K}{j + K}$ (the entire substring of $Y$ relevant at $v$).
\end{itemize}
Moreover, for each node $v$ in the partition tree we also define:
\begin{itemize}
\item A \emph{rate} (or \emph{additive accuracy}) $r_v$. This parameter serves two purposes. First, it gives a budget for the number of characters that can be read at $v$ (i.e., we are allowed to read~$r_v |X_v|$ symbols). Second, it determines the \emph{additive} approximation guarantee at $v$, in the sense that the output at $v$ is allowed to have additive error $1/r_v$. For the root node we set $r_v = 1000 / K$.
\item The \emph{multiplicative accuracy $\alpha_v > 1$.} This parameter determines the multiplicative approximation ratio at $v$. For the root node we set $\alpha_v = 10$.
\end{itemize}
We are now ready to define the precise computational task for each node:

\begin{definition}[Tree Distance Problem] \label{def:problem}
Given a node $v$ in the partition tree $T$, compute a list of numbers $\Delta_{v, -K}, \dots, \Delta_{v, K}$ such that
\begin{equation} \label{eq:delta-approx}
    \frac1{\alpha_v} \ED^{\leq K}(X_v, Y_{v, s}) - \frac1{r_v} \leq \Delta_{v, s} \leq \alpha_v \TD^{\leq K}_{v, s}(X, Y) + \frac1{r_v}.
\end{equation}
\end{definition}
Given a node $v$ of the partition tree, we will sometimes refer to solving the Tree Distance Problem at $v$ as simply \emph{solving $v$}.

As a sanity check, let us confirm that an algorithm for the Tree Distance Problem can distinguish edit distances with a polylogarithmic gap. Assume that we solve the Tree Distance Problem for the root node $v$, and let $\Delta = \Delta_{v, 0}$. Then, \Cref{eq:delta-approx} implies that \makebox{$0.1 \cdot \ED^{\leq K}(X, Y) - 0.001 K \leq \Delta \leq 10 \cdot \TD^{\leq K}(X, Y) + 0.001 K$}. Consequently, we can distinguish whether the edit distance $\ED(X, Y)$ is at most $k = K / (1000 B \log_B(n))$ or at least~$K$: On the one hand, if $\ED(X, Y) \leq k$ then \cref{lem:equivalence-ed-td,lem:capped-td-equivalence} imply that $\Delta \leq 0.02 K + 0.001 K = 0.021 K$. On the other hand, if $\ED(X, Y) \geq K$ then $\Delta \geq 0.099 K$. 

\paragraph{Second Ingredient: The Precision Sampling Lemma}
The next step is to assign appropriate sampling rates~$r_v$ to every node in the partition tree. The rate at any node indicates that we can afford to read an~$r_v$-fraction of the characters in the computation of that entire subtree (at the root, we will set $r_v = 1000/k$ so we read $\Order(n/k)$ characters in total). 

Focus on some node with rate $r_v$. Due to the sampling rate, we are bound to incur an \emph{additive} error of at least $1/r_v$. The challenge is now the following: We want to assign rates to the $B$ children of $v$ in such a way that we can obtain a good approximation of the tree distance at $v$ after combining the results from the children. Naively assigning the same rate to all children incurs an additive error of at least $B / r_v$, which is too large. To decrease the error we can increase the rate, but this results in reading more than the $r_v \cdot |X_v|$ characters we are aiming for.

Roughly speaking we have an instance of the following problem: There are unknown numbers $A_1, \dots, A_n \in \Real$. We can specify \emph{precisions} $u_1, \dots, u_n$ and obtain estimates $\widetilde A_i$ such that $|A_i - \widetilde A_i| \leq u_i$, where the \emph{cost} of each estimate is $1/u_i$ (in our setting the cost corresponds to the number of characters we read). The goal is to set the precisions appropriately to be able to distinguish whether $\sum_i A_i < 0.1$ or $\sum_i A_i > 10$, say, and minimize the total cost~$\sum_i 1/u_i$. If we set the precisions equally, we would need to have~$u_i < 10/n$ (otherwise we cannot distinguish the case where $A_i = 10/n$ for all $i$ from the case $A_i = 0$ for all $i$), which incurs in total cost $\Omega(n^2)$. Andoni, Krauthgamer and Onak~\cite{AndoniKO10} give a very elegant randomized solution to this problem with total cost $\widetilde\Order(n)$ and good error probability.

\begin{restatable}[Precision Sampling Lemma]{lemma}{lempsl}\label{lem:precision-sampling}
Let $\varepsilon, \delta > 0$. There is a distribution $\mathcal D = \mathcal D(\varepsilon, \delta)$ supported over the real interval $(0,1]$ with the following properties:
\begin{itemize}
\itemdesc{Accuracy:} Let $A_1, \dots, A_n \in \Real$ and let $u_1, \dots, u_n \sim {\cal D}$ be sampled independently. There is a recovery algorithm with the following guarantee: Given $(a, b \cdot u_i)$\=/approximations~$\widetilde A_i$ of~$A_i$, the algorithm $((1+\varepsilon) \cdot a, b)$\=/approximates $\sum_i A_i$ in time $\Order(n \cdot \varepsilon^{-2} \log(\delta^{-1}))$, with probability at least $1 - \delta$ and for any parameters $a \geq 1$ and $b \geq 0$.
\itemdesc{Efficiency:} Sample $u \sim {\cal D}$. Then, for any $N \geq 1$ there is an event $E = E(u)$ happening with probability at least $1 - 1/N$, such that $\Ex_{u \sim \cal D}[\,1/u \mid E\,] \leq \widetilde{\Order}(\varepsilon^{-2} \log(\delta^{-1}) \log N)$.
\end{itemize}
\end{restatable}

The Precision Sampling Lemma was first shown in~\cite{AndoniKO10} and later refined and simplified in~\cite{AndoniKO11,Andoni17}. For completeness, we give a full proof in~\cref{sec:psl}.

\paragraph{Third Ingredient: A Range Minimum Problem}
The final ingredient is an efficient algorithm to combine the recursively computed tree distances. Specifically, the following subproblem can be solved efficiently:

\begin{restatable}[A Range Minimum Problem]{lemma}{lemrangeminimum} \label{lem:range-minimum}
There is an $\Order(K)$-time algorithm for the following problem: Given integers $A_{-K}, \dots, A_K$, compute for all $s \in \set{-K, \dots, K}$:
\begin{equation*}
    B_s = \min_{-K \leq s' \leq K} A_{s'} + 2 \cdot |s - s'|.
\end{equation*}
\end{restatable}

In~\cite{AndoniKO10}, the authors use efficient Range Minimum queries (for instance implemented by segment trees) to \emph{approximately} solve this problem with a polylogarithmic overhead. We present a simpler, faster and \emph{exact} algorithm; see \cref{sec:range-minimum} for the pseudocode and proof.

\paragraph{The Framework}
We are ready to assemble the three ingredients in a baseline algorithm which solves the Tree Distance Problem; see \cref{alg:ako} for the pseudocode. We will first sketch the correctness of this algorithm. Later (in \cref{sec:overview:sec:ako-algorithm,sec:overview:sec:alg}) we will improve this algorithm to first obtain the Andoni-Krauthgamer-Onak algorithm, and second our sublinear-time version.

\begin{algorithm}[t]
\caption{The framework} \label{alg:ako}
\begin{algorithmic}[1]
\Input{Strings $X, Y$, a node $v$ in the partition tree $T$ and a sampling rate $r_v > 0$}
\Output{$\Delta_{v, s}$ for all shifts $s \in \set{-K, \dots, K}$}
\medskip
\If{$v$ is a leaf} \label{alg:ako:line:test-1-condition}
    \State\Return $\Delta_{v, s} = \ED(X_v, Y_{v, s})$ for all $s \in \set{-K, \dots, K}$ \label{alg:ako:line:test-1}
\EndIf
\ForEach{$i \in \rangezero B$} \label{alg:ako:line:iter-children}
    \State Let $v_i$ be the $i$-th child of $v$ and sample $u_{v_i} \sim \mathcal D$ (with appropriate parameters) \label{alg:ako:line:precision}
    \State Recursively compute $\Delta_{v_i, s}$ with rate $r_v / u_{v_i}$ for all $s \in \set{-K, \dots, K}$ \label{alg:ako:line:recursion}
    \State Compute $\widetilde A_{i, s} = \min_{s' \in \set{-K, \dots, K}} \Delta_{v_i, s'} + 2 \cdot |s - s'|$ using \cref{lem:range-minimum} \label{alg:ako:line:range-minimum}
\EndForEach
\ForEach{$s \in \set{-K, \dots, K}$} \label{alg:ako:line:iter-output}
    \State Let $\Delta_{v, s}$ be the result of the recovery algorithm (\cref{lem:precision-sampling}) applied to  \label{alg:ako:line:recovery}
    \Statex[1] $\widetilde A_{0, s}, \dots, \widetilde A_{B-1, s}$ with  precisions $u_{v_0}, \dots, u_{v_{B-1}}$
\EndForEach
\State\Return $\min(\Delta_{v, s}, K)$ for all $s \in \set{-K, \dots, K}$ \label{alg:ako:line:return}
\end{algorithmic}
\end{algorithm}

We quickly discuss the correctness of \cref{alg:ako}. In \cref{alg:ako:line:test-1-condition,alg:ako:line:test-1} we test whether the node $v$ can be solved trivially: If $v$ is a leaf, then we have reached the base case of \cref{def:tree-distance} and computing the tree distances boils down to computing the edit distance. In \crefrange{alg:ako:line:iter-children}{alg:ako:line:return} we use the approximations recursively computed by $v$'s children to solve~$v$. The idea is to approximately evaluate the following expression for the tree distance (which is equivalent to \Cref{eq:capped-td}):
\begin{equation*}
    \TD^{\leq K}_{v, s}(X, Y) = \min\left(\sum_{i \in \rangezero B} A_{i, s},\,K\right)\!, \quad A_{i, s} = \min_{-K \leq s' \leq K} (\TD^{\leq K}_{v_i, s'}(X, Y) + 2 \cdot |s - s'|).
\end{equation*}
For each child $v_i$ of $v$ the algorithm first recursively computes an approximation~$\Delta_{v_i, s'}$ of $\TD^{\leq K}_{v_i, s'}(X, Y)$ in \cref{alg:ako:line:recursion}. Then, in \cref{alg:ako:line:range-minimum}, we exactly evaluate
\begin{equation*}
    \widetilde A_{i, s} = \min_{-K \leq s' \leq K} \Delta_{v_i, s'} + 2 \cdot |s - s'|
\end{equation*}
for all $s$ using \cref{lem:range-minimum}. In the next step, the Precision Sampling Lemma comes into play. The recursive call returns approximations $\Delta_{v_i, s'}$ with multiplicative error $\alpha_{v_i}$ and additive error $1/r_{v_i} = u_{v_i} / r_v$. Hence, $\widetilde A_{i, s}$ is also an approximation of $A_{i, s}$ with multiplicative error~$\alpha_{v_i}$ and additive error~$u_{v_i} / r_v$. In this situation, the Precision Sampling Lemma (\cref{lem:precision-sampling}) allows to approximate the sum~$\sum_i A_{i, s}$ (for some fixed $s$) with multiplicative error $(1 + \varepsilon) \cdot \alpha_{v_i}$ and additive error $1/r_v$---just as required if we set $\alpha_v > (1 + \varepsilon) \cdot \alpha_{v_i}$. Hence, the values~$\Delta_{v, s}$ computed in \cref{alg:ako:line:recovery} are as claimed (up to taking the minimum with $K$ in order to obtain estimates for the capped tree distance, see \cref{alg:ako:line:return}).

In the next \cref{sec:overview:sec:ako-algorithm,sec:overview:sec:alg} we will instantiate this framework by adding \emph{pruning rules} to \cref{alg:ako}. That is, we design rules to directly solve certain nodes $v$ without having to recur on $v$'s children. In this way we will reduce the number of recursive calls and thereby the running time of \cref{alg:ako}.

\subsection{The Andoni-Krauthgamer-Onak Algorithm}\label{sec:overview:sec:ako-algorithm}
For completeness (and also to showcase how to instantiate the previous framework), we quickly demonstrate how a single pruning rule leads to the algorithm by Andoni, Krauthgamer and Onak~\cite{AndoniKO10}. Specifically, we add the following pruning rule to \cref{alg:ako}: \emph{If~$|X_v| \leq 1/r_v$, then return $\Delta_{v, s} = 0$ for all $s \in \set{-K, \dots, K}$.} This rule is correct, since the edit distance $\ED(X_v, Y_{v, s})$ is at most $|X_v|$. Hence, returning zero satisfies the condition in \Cref{eq:delta-approx}.

It remains to analyze the running time. In this regard the analysis differs quite substantially from our new sublinear-time algorithm. We will therefore be brief here and refer to \cref{sec:ako} for the complete proof. A single execution of \cref{alg:ako} (ignoring the recursive calls and lower-order factors such as $B$) takes roughly time $\Order(K)$. However, note that the recursive calls do not necessarily reach every node in the partition tree: Some nodes~$v$ are trivially solved by the new rule and thus their children are never explored. Let us call a node $v$ \emph{active} if the recursive computation reaches~$v$. One can bound the number of active nodes by $n / K \cdot (\log n)^{\Order(\log_B(n))}$ hence obtaining the time bound~$n \cdot (\log n)^{\Order(\log_B(n))}$.

\subsection{Our Algorithm} \label{sec:overview:sec:alg}
We are finally ready to describe the pruning rules leading to our sublinear-time algorithm. In contrast to the previous section, our pruning rules will allow us to bound the number of active nodes by $\poly(K)$. We will however spend more time for each active node: In the Andoni-Krauthgamer-Onak algorithm the running time per node is essentially~$K$, whereas in our version we run some more elaborate tests per node spending time proportional to~\makebox{$r_v |X_v| + \poly(K)$}. We will now progressively develop the pruning rules; the pseudocode is given at the end of this section.

\subsubsection{The Big Picture}
\paragraph{First Insight: Matching Substrings}
The first insight is that if $\ED(X, Y) \leq K$, then we can assume that for almost all nodes $v$ there exists a shift $s^* \in \set{-K, \dots, K}$ for which~\makebox{$X_v = Y_{v, s^*}$}. In this case, we say that the node~$v$ is \emph{matched}. The benefit is that if we know that $v$ is matched with shift $s^*$, then instead of approximating the edit distances between~$X_v$ and all shifts~$Y_{v,s}$, we can instead approximate the edit distance between $Y_{v, s^*}$ and all shifts~$Y_{v,s}$. That is, it suffices to compute the edit distances between a string and \emph{a shift of itself}.

To see that almost all nodes are matched, consider an optimal alignment between~$X$ and~$Y$ and recall that each level of the partition tree induces a partition of $X$ into substrings~$X_v$. The number of misaligned characters is bounded by $\ED(X, Y)$, thus there are at most $\ED(X, Y)$ parts~$X_v$ containing a misalignment. For all other parts, the complete part~$X_v$ is perfectly matched to some substring $Y_{v, s}$. We prove the claim formally in \cref{lem:alignment-many}. In light of this insight, we can assume that there are only $K$ unmatched nodes---otherwise, the edit distance and thereby also the tree distance between~$X$ and~$Y$ exceeds~$K$ and we can stop the algorithm. In the following we will therefore focus on matched nodes only.

For now we assume that we can efficiently test whether a node is matched. We justify this assumption soon (see the \emph{Matching Test} in~\cref{lem:alignment-test}) by giving a property tester for this problem in sublinear time.

\paragraph{Second Insight: Structure versus Randomness}
The second idea is to exploit a structure versus randomness dichotomy on strings: As the two extreme cases, a string is either periodic or random-like. The hope is that whenever~$Y_v$ falls into one of these extreme categories, then we efficiently approximate $\ED(Y_{v, s^*}, Y_{v, s})$ (without expanding~$v$'s children). Concretely, we use the following measure to interpolate between periodic and random-like:

\begin{definition}[Block Periodicity]
Let $Y$ be a string. The \emph{$K$-block periodicity $\bp_K(Y)$} of~$Y$ is the smallest integer $L$ such that $Y$ can be partitioned into $Y = \bigcirc_{\ell=1}^L Y_\ell$, where each substring~$Y_\ell$ is $K$-periodic (i.e., $Y_\ell$ is periodic with period length at most $K$).
\end{definition}

Suppose for the moment that we could efficiently compute the block periodicity of a string $Y$. Under this assumption, we first compute $\bp_{4K}(Y_v)$ for each matched node $v$ in the tree and distinguish three regimes depending on whether the block periodicity is small, large or intermediate. In the following section we discuss the pruning rules that we apply in these regimes.

\subsubsection{Pruning Rules} \label{sec:overview:sec:alg:sec:regimes}
\paragraph{The Periodic Regime: $\bp_{4K}(Y_v) = 1$}
For this case, we can approximate the edit distances via the following lemma, where you can think of $Y = Y_v$ and $Y_s = Y_{v, s}$.

\begin{restatable}[Periodic Rule]{lemma}{lemperiodicrule}\label{lem:periodicity-rule}
Let $Y$ be a string and write $Y_s = Y \range{K + s}{|Y| - K + s}$. If~$Y$ is periodic with primitive period $P$ and $|Y| \geq |P|^2 + 2K$, then for all~$s, s' \in \set{-K, \dots, K}$:
\begin{equation*}
    \ED(Y_s, Y_{s'}) = 2 \cdot \min_{j \in \Int} \big|\, s - s' + j|P| \,\big|.
\end{equation*}
\end{restatable}

Note that if some node $v$ that is matched and we know that $Y_v$ is periodic, then \cref{lem:periodicity-rule} allows us to return $\Delta_{v, s} = 2 \cdot \min_{j \in \Int} \big|\, s - s^* + j|P| \,\big|$ as the desired estimates for each shift~$s$. In this way we do not recur on $v$'s children and thereby prune the entire subtree below $v$.

To get some intuition for why \cref{lem:periodicity-rule} holds, observe that $2 \cdot \min_j \big|\, s - s' + j|P| \,\big|$ is exactly the cost of aligning both $Y_s$ and $Y_{s'}$ in such a way that all occurrences of the period~$P$ match (i.e., we shift both strings to the closest-possible occurrence of~$P$). The interesting part is to prove that this alignment is best-possible. We give the complete proof in \cref{sec:main:sec:rules}.

\paragraph{The Random-Like Regime: $\bp_{4K}(Y_v) > 10K$}
Next, we give the analogous pruning rule for the case when the block periodicity of $Y_v$ is large. We show that in this case, the best possible way to align any two shifts~$Y_{v, s}$ and~$Y_{v, s'}$ is to insert and delete $|s - s'|$ many characters. In this sense, $Y_v$ behaves like a \emph{random} string.

\begin{restatable}[Random-Like Rule]{lemma}{lemrandomrule}\label{lem:random-rule}
Let $Y$ be a string and write $Y_s = Y \range{K + s}{|Y| - K + s}$. If $\bp_{4K}(Y) > 10K$, then for all~$s, s' \in \set{-K, \dots, K}$:
\begin{equation*}
    \ED(Y_s, Y_{s'}) = 2 \cdot |s - s'|.
\end{equation*}
\end{restatable}

Again, this rule can be used to solve a matched node $v$ directly, pruning the subtree below $v$. We give the proof in \cref{sec:main:sec:rules}.

\paragraph{The Intermediate Regime: $1 < \bp_{4K}(Y_v) \leq 10K$}
We cannot directly approximate the edit distance $\ED(Y_{v, s^*}, Y_{v, s})$ in this case. Instead, we exploit the following lemma to argue that the branching procedure below $v$ is computationally cheap:

\begin{restatable}[Intermediate]{lemma}{lemintermediate} \label{lem:intermediate-efficient}
Let $v$ be a node in the partition tree. Then, in any level in the subtree below $v$, for all but at most $2\bp_{4K}(Y_v)$ nodes $w$ the string $Y_w$ is $4K$-periodic.
\end{restatable}

Indeed, since any node $v$ for which $Y_w$ is $4K$-periodic can be solved by the Periodic Rule, this lemma implies that there are at most $20K$ active nodes on any level in the partition subtree below any matched node $v$.

\subsubsection{String Property Testers} \label{sec:overview:sec:alg:sec:detecting-periodicity}
Next, we describe how to remove the assumptions that we can efficiently test whether a node is matched and that we can compute block periodicities. We start with the second task.

\paragraph{Computing Block Periodicities?}
The most obvious approach would be to show how to compute (or appropriately approximate) the block periodicity. This is indeed possible, but leads to a more complicated and slower algorithm (in terms of $\poly(K)$).

Instead, we twist the previous argument: We first show how to detect whether a string is periodic (in the straightforward way, see the following \cref{lem:periodicity-test}). For any matched node~$v$ we then run the following procedure: If $Y_v$ is $4K$-periodic, then we solve $v$ according to \cref{lem:periodicity-rule}. Otherwise, we continue to explore $v$'s children---with the following constraint: If at some point there are more than $20K$ active nodes which are not $4K$-periodic on any level in the recursion tree below $v$, then we interrupt the computation of this subtree and immediately solve $v$ according to \cref{lem:random-rule}. This approach is correct, since by \cref{lem:intermediate-efficient} witnessing more than $20K$ active nodes which are not $4K$-periodic on any level serves as a certificate that $\bp_{4K}(Y_v)$ is large. To test whether $Y_v$ is $4K$-periodic, we use the following tester:

\begin{restatable}[Periodicity Test]{lemma}{lemperiodicitytest} \label{lem:periodicity-test}
Let $X$ be a string, and let $r > 0$ be a sampling rate. There is an algorithm which returns one of the following two outputs:
\begin{itemize}
\item \Close[P], where $P$ is a primitive string of length $\leq K$ with $\HD(X, P^* \range{0}{|X|}) \leq 1/r$.
\item \Far, in which case $X$ is not $K$-periodic.
\end{itemize}
The algorithm runs in time~$\Order(r |X| \log(\delta^{-1}) + K)$ and is correct with probability $1 - \delta$.
\end{restatable}

Note that as we are shooting for a sublinear-time algorithm we have to resort to a property tester which can only distinguish between \emph{close} and \emph{far} properties (in this case: periodic or \emph{far from periodic}). The proof of \cref{lem:periodicity-test} is simple, see \cref{sec:main:sec:property-tests} for details.

\paragraph{Testing for Matched Nodes}
We need another algorithmic primitive to test whether a node is matched. Again, since our goal is to design a sublinear-time algorithm, we settle for the following algorithm which distinguishes whether $v$ is matched or \emph{far from matched}.

\begin{restatable}[Matching Test]{lemma}{lemalignmenttest} \label{lem:alignment-test}
Let $X, Y$ be strings such that $|Y| = |X| + 2K$, and let $r > 0$ be a sampling rate. There is an algorithm which returns one of the following two outputs:
\begin{itemize}
\item \Close[s^*], where $s^* \in \set{-K, \dots, K}$ satisfies $\HD(X, Y \range{K + s^*}{|X| + K + s^*}) \leq 1/r$.
\item \Far, in which case there is no $s^* \in \set{-K, \dots, K}$ with $X = Y \range{K + s^*}{|X| + K + s^*}$.
\end{itemize}
The algorithm runs time $\Order(r |X| \log(\delta^{-1}) + K \log |X|)$ and is correct with probability $1 - \delta$.
\end{restatable}

The proof of \cref{lem:alignment-test} is non-trivial. It involves checking whether $X$ and $Y$ follow a common period---in this case we can simply return any shift respecting the periodic pattern. If instead we witness errors to the periodic pattern, then we try to identify a shift under which also the errors align. See \cref{sec:main:sec:property-tests} for the details.

\subsubsection{Putting the Pieces Together} \label{sec:overview:sec:alg:sec:assemble}
We finally assemble our complete algorithm. The pseudocode is given in \cref{alg:main}. In this section we will sketch that \cref{alg:main} correctly and efficiently solves the Tree Distance Problem. The formal analysis is deferred to~\cref{sec:main}.

\begin{algorithm}[t]
\caption{} \label{alg:main}
\begin{algorithmic}[1]
\Input{Strings $X, Y$, a node $v$ in the partition tree $T$ and a rate $r_v$}
\Output{$\Delta_{v, s}$ for all shifts $s \in \set{-K, \dots, K}$}
\medskip
\If{$v$ is a leaf \OR{} $|X_v| \leq 100 K^2$} \label{alg:main:line:test-short-condition}
    \State\Return $\Delta_{v, s} = \ED^{\leq K}(X_v, Y_{v, s})$ (or compute $2$-approximations using \cref{thm:ed-shifts}) \label{alg:main:line:test-short}
\EndIf
\State Run the Matching Test (\cref{lem:alignment-test}) for $X_v, Y_v$ (with $r = 3r_v$ and $\delta = 0.01 \cdot K^{-100}$) \label{alg:main:line:alignment-test}
\State Run the $4K$-Periodicity Test (\cref{lem:periodicity-test}) for $Y_v$ (with $r = 3r_v$ and $\delta = 0.01 \cdot K^{-100}$) \label{alg:main:line:periodicity-test}
\If{the Matching Test returns \Close[s^*]} \label{alg:main:line:alignment-test-condition}
    \If{the Periodicity Test returns \Close[P]} \label{alg:main:line:periodicity-test-condition}
        \State\Return $\Delta_{v, s} = 2 \cdot \min_{j \in \Int} |\,s - s^* + j |P|\,|$ \label{alg:main:line:periodicity-rule}
    \Else
        \State Continue in \cref{alg:main:line:iter-children} with the following exception: If at some point during the \label{alg:main:line:random-rule}
        \Statex[2] recursive computation there is some level containing more than $20K$ active 
        \Statex[2] nodes below $v$ for which the Periodicity Test (in \cref{alg:main:line:periodicity-test}) reports \Far, then
        \Statex[2] interrupt the recursive computation and \Return $\Delta_{v, s} = 2 \cdot |s - s^*|$
    \EndIf
\EndIf

\medskip
\ForEach{$i \in \rangezero B$} \label{alg:main:line:iter-children}
    \State Let $v_i$ be the $i$-th child of $v$ and sample $u_{v_i} \sim \mathcal D((200 \log K)^{-1}, 0.01 \cdot K^{-101})$
    \State Recursively compute $\Delta_{v_i, s}$ with rate $r_v / u_{v_i}$ for all $s \in \set{-K, \dots, K}$ \label{alg:main:line:recursion}
    \State Compute $\widetilde A_{i, s} = \min_{s' \in \set{-K, \dots, K}} \Delta_{v_i, s'} + 2 \cdot |s - s'|$ using \cref{lem:range-minimum} \label{alg:main:line:range-minimum}
\EndForEach
\ForEach{$s \in \set{-K, \dots, K}$}
    \State Let $\Delta_{v, s}$ be the result of the recovery algorithm (\cref{lem:precision-sampling}) applied to  \label{alg:main:line:recovery}
    \Statex[1] $\widetilde A_{0, s}, \dots, \widetilde A_{B-1, s}$ with  precisions $u_{v_0}, \dots, u_{v_{B-1}}$
\EndForEach
\State\Return $\min(\Delta_{v, s}, K)$ for all $s \in \set{-K, \dots, K}$ \label{alg:main:line:return}
\end{algorithmic}
\end{algorithm}

\paragraph{Correctness}
We first sketch the correctness of \cref{alg:main}. The recursive calls in \crefrange{alg:main:line:iter-children}{alg:main:line:return} are essentially copied from \cref{alg:ako} (except for differences in the parameters which are not important here) and correct by the same argument as before (using the Precision Sampling Lemma as the main ingredient). The interesting part happens in \crefrange{alg:main:line:test-short-condition}{alg:main:line:random-rule}. In \cref{alg:main:line:test-short-condition,alg:main:line:test-short} we test whether the strings are short enough so that we can afford to compute the edit distances $\ED(X_v, Y_{v, s})$ by brute-force. If not, we continue to run the Matching Test for $X_v$ and $Y_v$ and the Periodicity Test for $Y_v$.

There are two interesting cases---both assume that the Matching Test reports \Close[s^*] and therefore~$X_v \approx Y_{v, s^*}$ where $\approx$ denotes equality up to $1/(3r_v)$ Hamming errors. If also the Periodicity Test returns \Close[P] then we are in the situation that $X_v \approx Y_{v, s^*} \approx P^*$. Moreover~$Y_{v, s}$ is $\approx$-approximately equal to a shift of $P^*$. \cref{lem:periodicity-rule} implies that the edit distance between~$X_v$ and $Y_{v, s}$ is approximately $2 \cdot \min_{j \in \Int} |\,s - s^* + j|P|\,|$. We lose an additive error of $1/(3r_v)$ for each of the three $\approx$ relations, hence the total additive error is $1/r_v$ as hoped, and we do not introduce any multiplicative error.

Next, assume that the Periodicity Test reports \Far. Then, as stated in \cref{alg:main:line:random-rule} we continue the recursive computation, but with an exception: If at some point during the recursive computation there are more than $20K$ active nodes on any level for which the Periodicity Test reports \Far, then we interrupt the computation and return $\Delta_{v, s} = 2 \cdot |s - s^*|$. Suppose that indeed this exception occurs. Then there are more than $20K$ nodes $w$ on one level of the partition tree below~$v$ for which $Y_w$ is \emph{not} $4K$-periodic. Using \cref{lem:intermediate-efficient} we conclude that~$\bp_{4K}(Y_v) > 10K$, and therefore returning $\Delta_{v, s} = 2 \cdot |s - s^*|$ is correct by \cref{lem:random-rule}.

\paragraph{Running Time}
We will first think of running \cref{alg:main} with $T$ being a balanced $B$-ary partition tree with $n$ leaves, just as as in the Andoni-Krauthgamer-Onak algorithm (we will soon explain why we need to modify this). To bound the running time of \cref{alg:main} we first prove that the number of active nodes in the partition tree is bounded by $\poly(K)$. One can show that there are only $\poly(K)$ unmatched nodes in the tree (see \cref{lem:alignment-many}), so we may only focus on the matched nodes. Each matched node, however, is either solved directly (in \cref{alg:main:line:test-short} or in \cref{alg:main:line:periodicity-rule}) or continues the recursive computation with at most $\poly(K)$ active nodes (in \cref{alg:main:line:random-rule}). In \cref{sec:main} we give more details.

Knowing that the number of active nodes is small, we continue to bound the total expected running time of \cref{alg:main}. It is easy to check that a single execution of \cref{alg:main} (ignoring the cost of recursive calls) is roughly in time $r_v |X_v| + \poly(K)$. Using the Precision Sampling Lemma, we first bound $r_v$ by $1/K \cdot (\log n)^{\Order(\log_B(n))}$ in expectation, for any node~$v$. Therefore, the total running time can be bounded by
\begin{align*}
    \sum_{\text{$v$ active}} (r_v |X_v| + \poly(K))
    &\leq \frac1K \cdot (\log n)^{\Order(\log_B(n))} \cdot \sum_v |X_v| + \poly(K) \\
    &\leq \frac nK \cdot (\log n)^{\Order(\log_B(n))} + \poly(K).
\end{align*}
Here we used that $\sum_w |X_w| = n$ where the sum is over all nodes~$w$ on any fixed level in the partition tree. It follows that $\sum_v |X_v| \leq n \log n$, summing over \emph{all} nodes $v$.

\paragraph{Optimizing the Lower-Order Terms}
This running time bound does not match the claimed bound in \cref{thm:main}: The overhead $(\log n)^{\Order(\log_B(n))}$ should rather be $(\log K)^{\Order(\log_B(K))}$ and not depend on $n$. If $n \leq K^{100}$, say, then both terms match. But we can also reduce the running time in the general case by ``cutting'' the partition tree at depth $\log_B(K^{100})$. That is, we delete all nodes below that depth from the partition tree and treat the nodes at depth $\log_B(K^{100})$ as leaves in the algorithm. The remaining tree is still a partition tree, according to \cref{def:partition-tree}. The correctness argument remains valid, but we have to prove that the running time of \cref{alg:main:line:test-short} does not explode. For each leaf at depth~$\log_B(K^{100})$ we have that~$|X_v| \leq n / K^{100}$ and for that reason computing $\ED^{\leq K}(X_v, Y_{v, s})$ for all shifts $s$ (say with the Landau-Vishkin algorithm) takes time $\Order(n / K^{99} + K^3)$. Since there are much less than~$K^{99}$ active nodes, the total contribution can again be bounded by~$\Order(n / K + \poly(K))$.

\paragraph{Optimizing the Polynomial Dependence on $K$}
Finally, let us pinpoint the exponent of the polynomial dependence on $K$. In the current algorithm we can bound the number of active nodes by roughly $K^2$ (there are at most $K$ nodes per level for which the matching test fails, and the subtrees rooted at these nodes contain at most $K$ active nodes per level). The most expensive step in \cref{alg:main} turns out to be the previously mentioned edit distance computation in \cref{alg:main:line:test-short}. Using the Landau-Vishkin algorithm (with cap~$K$) for~$\Order(K)$ shifts~$s$, the running time of \cref{alg:main:line:test-short} incurs a cubic dependence on $K$, and therefore the total dependence on $K$ becomes roughly $K^5$.

We give a simple improvement to lower the dependence to $K^4$ and leave further optimizations of the $\poly(K)$ dependence as future work. The idea is to use the following result on approximating the edit distances \emph{for many shifts $s$:}

\begin{restatable}[Edit Distance for Many Shifts]{theorem}{thmedshifts} \label{thm:ed-shifts}
Let $X, Y$ be strings with $|Y| = |X| + 2K$. We can compute a (multiplicative) $2$-approximation of $\ED^{\leq K}(X, Y \range{K + s}{|X| + K + s})$, for all shifts~$s \in \set{-K, \dots, K}$, in time $\Order(|X| + K^2)$.
\end{restatable}

This result can be proven by a modification of the Landau-Vishkin algorithm~\cite{LandauV88}; we provide the details in \cref{sec:ed-shifts}. It remains to argue that computing a $2$-approximation in \cref{alg:main:line:test-short} does not mess up the correctness proof. This involves the multiplicative approximation guarantee of our algorithm which we entirely skipped in this overview, and for this reason we defer the details to \cref{sec:main}.

\paragraph{Implementation Details}
We finally describe how to implement the interrupt condition in \cref{alg:main:line:random-rule}. Note that the order of the recursive calls in \cref{alg:main:line:recursion} is irrelevant. In the current form the algorithm explores the partition tree in a depth-first search manner, but we might as well use breadth-first search. For any node $v$ which reaches \cref{alg:main:line:random-rule} we may therefore continue to compute all recursive computations using breadth-first search. If at some point we encounter one search level containing more than $200K$ active nodes for which the Periodicity Test reports \Far, we stop the breadth-first search and jump back to \cref{alg:main:line:random-rule}. With this modification the algorithm maintains one additional counter which does not increase the asymptotic time complexity.

\section{Our Algorithm in Detail} \label{sec:main}
In this section we give the formal analysis of \cref{alg:main}. We split the proof into the following parts: In \cref{sec:main:sec:periodicity} we prove some lemmas about periodic and block-periodic strings. In \cref{sec:main:sec:rules} we give the structural lemmas about edit distances in special cases. In \cref{sec:main:sec:property-tests} we prove the correctness of the string property testers (the ``Matching Test'' and ``Periodicity Test'').  In \cref{sec:main:sec:assemble} we finally carry out the correctness and running time analyses for \cref{alg:main}, and in \cref{sec:main:sec:main-theorem} we give a formal proof of our main theorem.

In the following proofs we will often use the following simple proposition.

\begin{proposition}[Alignments Have Small Stretch] \label{prop:alignment-stretch}
Let $X, Y$ be strings of equal length. If~$A$ is an optimal alignment between $X$ and $Y$, then $|i - A(i)| \leq \frac12 \ED(X, Y)$ for all~$0 \leq i \leq |X|$.
\end{proposition}
\begin{proof}
Since $A$ is an optimal alignment between $X$ and $Y$, we can write the edit distance $\ED(X, Y)$ as the sum $\ED(X \range 0i, Y \range0{A(i)}) + \ED(X \range i{|X|}, Y \range{A(i)}{|Y|})$. Both edit distances are at least $|i - A(i)|$ which is the length difference of these strings, respectively. It follows that $\ED(X, Y) \geq 2 \cdot |i - A(i)|$, as claimed.
\end{proof}

\subsection{Some Facts about Periodicity} \label{sec:main:sec:periodicity}
We prove the following two lemmas, both stating roughly that \emph{if a string $X$ closely matches a shift of itself, then $X$ is close to periodic.} The first lemma is easy and well-known. The second lemma is new.

\begin{lemma}[Self-Alignment Implies Periodicity] \label{lem:shift-periodic}
Let $X$ be a string. For any shift~$s > 0$, if $X \range{0}{|X| - s} = X \range{s}{|X|}$ then $X$ is $s$-periodic (with period $X \range{0}{s}$).
\end{lemma}
\begin{proof}
By assumption we have that $X[j] = X[s + j]$ for each $j \in \rangezero{|X|-s}$. It follows that for $P = X \range 0s$ we have $X \access j = P \access{j \bmod s}$ for all indices $j \in \rangezero{|X|}$ and thus $X = P^* \range0{|X|}$.
\end{proof}

\begin{lemma}[Self-Alignment Implies Small Block Periodicity] \label{lem:shift-block-periodic}
Let $X$ be a string. For any shift $s > 0$, if $\ED(X \range{0}{|X| - s}, X \range{s}{|X|}) < 2s$ then $\bp_{2s}(X) \leq 4s$.
\end{lemma}
\begin{proof}
Let $Y = X \range{0}{|X| - s}$, $Z = X \range{s}{|X|}$ and let $A$ denote an optimal alignment between~$Y$ and~$Z$. We will greedily construct a sequence of indices $0 = i_0 < \dots < i_L = |Y|$ as follows: Start with $i_0 = 0$. Then, having assigned $i_\ell$ we next pick the smallest index~$i_{\ell+1} > i_\ell$ for which~$Y \range{i_\ell}{i_{\ell+1}} \neq Z \range{A(i_\ell)}{A(i_{\ell+1})}$. Using that $A$ is an optimal alignment, we have constructed a sequence of $L < 2s$ indices, since
\begin{equation*}
    2s > \ED(Y, Z) = \sum_{\ell=0}^{L-1} \ED(Y \range{i_\ell}{i_{\ell+1}}, Z \range{A(i_\ell)}{A(i_{\ell+1})}) \geq L.
\end{equation*}
Moreover, by \cref{prop:alignment-stretch} we have that $|i - A(i)| < s$ for all $i$. The greedy construction guarantees that $Y \range{i_\ell}{i_{\ell+1} - 1} = Z \range{A(i_\ell)}{A(i_{\ell+1}-1)}$. Therefore, and since $Y, Z$ are substrings of~$X$, we have $X \range{i_\ell}{i_{\ell+1} - 1} = X \range{s + A(i_\ell)}{s + A(i_{\ell+1}-1)}$. We will now apply \cref{lem:shift-periodic} to these substrings of $X$ with shift $s' = s + A(i_\ell) - i_\ell$. Note that $0 < s' < 2s$ (which satisfies the precondition of \cref{lem:shift-periodic}) and thus $X \range{i_\ell}{i_{\ell+1} - 1}$ is $2s$-periodic.

Finally, consider the following partition of $X$ into $2L+1$ substrings 
\begin{equation*}
    X = \left(\bigcirc_{\ell=0}^{L-1} X \range{i_\ell}{i_{\ell+1} - 1} \circ X \access{i_{\ell+1} - 1}\right) \circ X \range{|X| - s}{|X|}.
\end{equation*}
We claim that each of these substrings is $2s$-periodic: For $X \range{i_\ell}{i_{\ell+1}-1}$ we have proved this in the previous paragraph, and the strings $X \access{i_{\ell+1} - 1}$ and $X \range{|X| - s}{|X|}$ have length less than~$2s$ and are thus trivially $2s$-periodic. This decomposition certifies that $\bp_{2s}(X) \leq 2L + 1 \leq 4s$.
\end{proof}

\subsection{Edit Distances between Periodic and Random-Like Strings} \label{sec:main:sec:rules}

The goal of this section is to prove the structural \cref{lem:periodicity-rule,lem:random-rule} which determine the edit distance between certain structured strings.

\lemperiodicrule*
\begin{proof}
For simplicity set $\Delta = 2 \cdot \min_{j \in \Int} \big| s - s' + j|P|\big|$ and $p = |P|$. We will argue that~$\Delta$ is both an upper bound and lower bound for $\ED(Y_s, Y_{s'})$. The upper bound is simple: Note that due to the periodicity of $Y$, we can transform $Y_{s'}$ into~$Y_{s}$ by deleting and inserting $\min_{j \in \Int} \big|\,s - s' + j|P|\,\big|$ many characters. Thus, $\ED(Y_s, Y_{s'}) \leq \Delta$.

Next, we prove the lower bound. For contradiction suppose that $\ED(Y_s, Y_{s'}) < \Delta$. We assumed that $|Y_s| = |Y| - 2K \geq p^2$, and we can therefore split $Y_s$ into $p$ parts of length $p$ plus some rest. Let $i_\ell = \ell \cdot p$ for all $i \in \set{0, \dots, p}$ and let $i_{p+1} = |Y_s|$; we treat $Y_s \range{i_\ell}{i_{\ell+1}}$ as the $\ell$-th part. Let~$A$ denote an optimal alignment between $Y_s$ and $Y_{s'}$; we have
\begin{equation*}
    \ED(Y_s, Y_{s'}) = \sum_{\ell=0}^p \ED(Y_s \range{i_\ell}{i_{\ell+1}}, Y_{s'} \range{A(i_\ell)}{A(i_{\ell+1})}).
\end{equation*}
We assumed that $\ED(Y_s, Y_{s'}) < \Delta \leq |P|$ (the latter inequality is by the definition of $\Delta)$ and therefore at least one of the first $p$ summands must be zero, say the $\ell$-th one, $\ell < p$. It follows that the two strings $Y_s \range{i_\ell}{i_{\ell+1}}$ and $Y_{s'} \range{A(i_\ell)}{A(i_{\ell+1})}$ are equal. Both are length-$p$ substrings of $Y$ and thus rotations of the global period $P$. We assumed that $P$ is primitive (i.e., $P$ is not equal to any of its non-trivial rotations) and therefore $s + i_\ell \equiv s' + A(i_\ell) \mod p$. By the definition of $\Delta$ we must have that $|A(i_\ell) - i_\ell| \geq \Delta / 2$. But this contradicts \cref{prop:alignment-stretch} which states that $|A(i_\ell) - i_\ell| \leq \ED(Y_s, Y_{s'}) / 2 < \Delta / 2$.
\end{proof}

\lemrandomrule*
\begin{proof}
Let $\Delta = 2 \cdot |s-s'|$. First note that $\ED(Y_s, Y_{s'}) \leq \Delta$ since we can transform $Y_s$ into $Y_{s'}$ by simply inserting and deleting $|s - s'|$ symbols. For the lower bound, suppose that~$\ED(Y_s, Y_{s'}) < \Delta$. Therefore, we can apply \cref{lem:shift-block-periodic} for an appropriate substring of~$Y$: Assume without loss of generality that $s \leq s'$ and let $Z = Y\range{K+s}{|Y|-K+s'}$. Then clearly $Y_s = Z \range{0}{|Z|-s'+s}$ and $Y_{s'} = Z \range{s'-s}{|Z|}$, so~\cref{lem:shift-block-periodic} applied to~$Z$ with shift~$s'-s$ yields that $\bp_{2(s'-s)}(Y') \leq 4 \cdot (s'-s)$. Using that $s' - s \leq 2K$ we conclude that $\bp_{4K}(Y') \leq 8K$. We can obtain $Y$ by adding at most $K$ characters to the start and end of $Z$. It follows that $\bp_{4K}(Y) \leq \bp_{4K}(Z) + 2 \leq 10K$. This contradicts the assumption in the lemma statement, and therefore $\ED(Y_s, Y_{s'}) \geq \Delta$.
\end{proof}

\lemintermediate*

\begin{proof}
Focus on some level of the computation subtree below $v$. We will bound the number of nodes $w$ in this level for which $Y_w$ is not $4K$-periodic. Note that if $Y_v$ was partitioned into~$Y_v = \bigcirc_w Y_w$, then by the definition of block periodicity we would immediately conclude that at most $\bp_{4K}(Y_v)$ many parts $Y_w$ are not $4K$-periodic. However, recall that for a node~$w$ with associated interval $I_w = \range{i}{j}$, we defined $Y_w$ as $Y_w = Y\range{i - K}{j + K}$. This means that the substrings $Y_w$ overlap with each other and therefore do \emph{not} partition $Y_v$. 

To deal with this, note that we can assume that $|Y_w| > 4K$, since otherwise $Y_w$ is trivially $4K$-periodic. Hence, each $Y_w$ can overlap with at most two neighboring nodes (since the intervals $I_w$ are disjoint). Therefore, we can divide the $w$'s in two groups such that the $Y_w$'s in each group do not not overlap with each other. For each group, we apply the argument from above to derive that there are at most $\bp_{4K}(Y_v)$ many $Y_w$'s which are not $4K$-periodic. In this way, we conclude that there are at most $2 \bp_{4K}(Y_v)$ nodes in the level which are not $4K$-periodic, as desired.
\end{proof}

\subsection{Some String Property Testers} \label{sec:main:sec:property-tests}
The main goal of this section is to formally prove \cref{lem:alignment-test,lem:periodicity-test}, that is, the Matching Test and Periodicity Test. As a first step, we need the following simple lemma about testing equality of strings.

\begin{lemma}[Equality Test] \label{lem:equality-test}
Let $X, Y$ be two strings of the same length, and let $r > 0$ be a sampling rate. There is an algorithm which returns of the following two outputs:
\begin{itemize}
\item \Close, in which case $\HD(X, Y) \leq 1/r$.
\item \Far[i], in which case $X \access i \neq Y \access i$.
\end{itemize}
The algorithm runs in time $\Order(r |X| \log(\delta^{-1}))$ and is correct with probability $1 - \delta$.
\end{lemma}
\begin{proof}
The idea is standard: For $r |X| \ln(\delta^{-1})$ many random positions $i \in \rangezero{|X|}$, test whether $X \access i = Y \access i$. If no error is found, then we report \Close. This equality test is clearly sound: If~$X = Y$, then it will never fail. It remains to argue that if $\HD(X, Y) > 1/r$ then the test fails with probability at least $1 - \delta$. Indeed, each individual sample finds a Hamming error with probability $(r |X|)^{-1}$. Hence, the probability of not finding any Hamming error across all samples is at most
\begin{equation*}
    \left(1 - \frac{1}{r |X|}\right)^{r |X| \ln(\delta^{-1})} < \exp(-\ln(\delta^{-1})) = \delta.
\end{equation*}
The running time is bounded by $\Order(r |X| \log(\delta^{-1}))$.
\end{proof}

\lemperiodicitytest*
\begin{proof}
We start analyzing the length-$2K$ prefix $Y = X \range{0}{2K}$. In time $\Order(K)$ we can compute the smallest period~$P$ such that $Y = P^* \range{0}{|Y|}$ by searching for the first match of $Y$ in~$Y \circ Y$, e.g.\ using the Knuth-Morris-Pratt pattern matching algorithm~\cite{KnuthMP77}. If no such match exists, we can immediately report~\Far. So suppose that we find a period~$P$. It must be primitive (since it is the smallest such period) and it remains to test whether~$X$ globally follows the period. For this task we use the Equality Test (\cref{lem:equality-test}) with inputs~$X$ and $P^*$ (of course, we cannot write down the infinite string~$P^*$, but we provide oracle access to~$P^*$ which is sufficient here). On the one hand, if~$X$ is indeed periodic with period~$P$, then the Equality Test reports \Close. On the other hand, if $X$ is $1/r$-far from any periodic string, then it particular $\HD(X, P^*) > 1/r$ and therefore the Equality Test reports \Far. The only randomized step is the Equality Test. We therefore set the error probability of the Equality Test to $\delta$ and achieve total running time $\Order(r |X| \log(\delta^{-1}) + K)$.
\end{proof}

\lemalignmenttest*
\begin{proof}
For convenience, we write $Y_s = Y \range{K + s}{|X| + K + s}$. Our goal is to obtain a single \emph{candidate shift $s^*$} (that is, knowing $s^*$ we can exclude all other shifts from consideration). Having obtained a candidate shift, we can use the Equality Test (\cref{lem:equality-test} with parameters~$r$ and $\delta/3$) to verify whether we indeed have $X = Y_{s^*}$. In the positive case, \cref{lem:equality-test} implies that $\HD(X, Y_{s^*}) \leq 1/r$, hence returning $s^*$ is valid. The difficulty lies in obtaining the candidate shift. Our algorithm proceeds in three steps:
\begin{enumerate}[itemsep=\smallskipamount]
\itemdesc{Aligning the Prefixes:} We start by computing the set $S$ consisting of all shifts $s$ for which $X \range{0}{2K} = Y_s \range{0}{2K}$. One way to compute this set in linear time $\Order(K)$ is by using a pattern matching algorithm with pattern $X \range{0}{2K}$ and text $Y \range{0}{4K}$ (like the Knuth-Morris-Pratt algorithm~\cite{KnuthMP77}). It is clear that~$S$ must contain any shift~$s$ for which globally $X = Y_s$. For that reason we can stop if $|S| = 0$ (in which case we return~\Far) or if $|S| = 1$ (in which case we test the unique candidate shift $s^* \in S$ and report accordingly).
\itemdesc{Testing for Periodicity:} After running the previous step we can assume that $|S| \geq 2$. Take any elements $s < s'$ from $S$; we have that $X \range{0}{2K} = Y_s \range{0}{2K} = Y_{s'} \range{0}{2K}$. It follows that $X \range{0}{2K - s' + s} = X \range{s' - s}{2K}$, and thus by \cref{lem:shift-periodic} we conclude that $X \range{0}{2K}$ is periodic with period $P = X \range{0}{s' - s}$, where $|P| \leq s' - s \leq 2K$. Obviously the same holds for $Y_s \range{0}{2K}$ and $Y_{s'} \range{0}{2K}$.

We will now test whether $X$ and $Y_s$ are also globally periodic with this period~$P$. To this end, we apply the Equality Test two times (each time with parameters~$2r$ and~$\delta/3$) to check whether $X = P^* \range{0}{|X|}$ and $Y_s = P^* \range{0}{|Y_s|}$. If both tests return \Close, then \cref{lem:equality-test} guarantees that~$\HD(X, P^* \range{0}{|X|}) \leq 1/(2r)$ and \makebox{$\HD(Y_s, P^* \range{0}{|Y_s|}) \leq 1/(2r)$} and hence, by the triangle inequality, $\HD(X, Y_s) \leq 1/r$. Note that we have witnessed a matching shift $s^* = s$.
\itemdesc{Aligning the Leading Mismatches:} Assuming that the previous step did not succeed, one of the Equality Tests returned \Far[i_0] for some position $i_0 > 2K$ with~$X \access{i_0} \neq P^* \access{i_0}$ or~$Y_s \access{i_0} \neq P^* \access{i_0}$. Let us refer to these indices as \emph{mismatches}. Moreover, we call a mismatch $i$ a \emph{leading mismatch} if the $2K$ positions to the left of~$i$ are not mismatches. We continue in two steps: First, we find a leading mismatch. Second, we turn this leading mismatch into a candidate shift.
\begin{enumerate}[itemsep=\smallskipamount, topsep=\smallskipamount, label=3\alph*]
\itemdesc{Finding a Leading Mismatch:} To find a leading mismatch, we use the following binary search-style algorithm: Initialize $L \gets 0$ and $R \gets i_0$. We maintain the following two invariants:~(i)~All positions in $\range{L}{L + 2K}$ are not mismatches, and (ii) $R$ is a mismatch. Both properties are initially true. We will now iterate as follows: Let $M \gets \ceil{(L + R) / 2}$ and test whether there is a mismatch~$i \in \range{M}{M + 2K}$. If there is such a mismatch~$i$, we update $R \gets i$. Otherwise, we update $L \gets M$. It is easy to see that in both cases both invariants are maintained. Moreover, this procedure is guaranteed to make progress as long as $L + 4K < R$. If at some point~$R \leq L + 4K$, then we can simply check all positions in $\range{L}{R}$---one of these positions must be a leading mismatch $i$.
\itemdesc{Finding a Candidate Shift:} Assume that the previous step succeeded in finding a leading mismatch $i$. Then we can produce a single candidate shift as follows: Assume without loss of generality that $X \access i \neq P^* \access i$, and let $i \leq j$ be the smallest position such that $Y_s \access j \neq P^* \access j$. Then $s^* = s + j - i$ is the only candidate shift (if it happens to fall into the range $\set{-K, \dots, K}$).

Indeed, for any $s'' > s^*$ we can find a position where $X$ and $Y_{s''}$ differ. To see this, we should assume that $s''$ respects the period (i.e., $P^* = P^* \range{K + s''}{\infty}$), since otherwise we find a mismatch in the length-$2K$ prefix. But then
\begin{align}
    Y_{s''} \access{j + s - s''} &= Y_s \access j \label{lem:matching-test:eq:1} \\
    &\neq P^* \access j \label{lem:matching-test:eq:2} \\
    &= P^* \access{j + s - s''} \label{lem:matching-test:eq:3} \\
    &= X \access{j + s - s''}, \label{lem:matching-test:eq:4}
\end{align}
which proves that $X \neq Y_{s''}$ and thereby disqualifies $s''$ as a feasible shift. Here we used \eqref{lem:matching-test:eq:1} the definition of $Y_s$, \eqref{lem:matching-test:eq:2} the assumption that $Y_s \access j \neq P^* \access j$, \eqref{lem:matching-test:eq:3} the fact that both $s$ and $s''$ respect the period $P$ and \eqref{lem:matching-test:eq:4} the assumption that $i$ was a leading mismatch which implies that $X$ matches $P^*$ at the position $j + s - s'' < i$.

A similar argument works for any shift $s'' < s^*$. In this case one can show that~$X \access i \neq P^* \access i = Y_{s''} \access i$ which also disqualifies $s''$ as a candidate shift.
\end{enumerate}
\end{enumerate}
We finally bound the error probability and running time of this algorithm. We only use randomness when calling the Equality Test which runs at most three times. Since each time we set the error parameter to $\delta / 3$, the total error probability is $\delta$ as claimed. The running time of the Equality Tests is bounded by $\Order(r |X| \log(\delta^{-1}))$ by \cref{lem:equality-test}. In addition, steps~1 and 2 take time $\Order(K)$. Step~3 iterates at most $\log |X|$ times and each iteration takes time~$\Order(K)$. Thus, the total running time is $\Order(r |X| \log(\delta^{-1}) + K \log|X|)$.
\end{proof}

\subsection{Putting The Pieces Together} \label{sec:main:sec:assemble}
\paragraph{Setting the Parameters}
Throughout this section we assume that $T$ is a balanced $B$-ary partition tree with $\min(n, K^{100})$ leaves, where each leaf $v$ is labeled with an interval $I_v$ of length $|I_v| \approx \max(1, n / K^{100})$. In particular there are at most $2 \cdot K^{100}$ nodes in the tree and its depth is bounded by $\ceil{\log_B \min(n, K^{100})}$. We also specify the following parameters for every node~$v$ in the partition tree:
\begin{itemize}
\item Rate $r_v$: If $v$ is the root then we set $r_v = 1000/K$. Otherwise, if $v$ is a child of $w$ then sample~$u_v \sim \mathcal D((200 \log K)^{-1}, K^{-200})$ independently and set $r_v = r_w / u_v$. (This assignment matches the values in \cref{alg:main}.) 
\item Multiplicative accuracy $\alpha_v = 10 \cdot (1 - (200 \log K)^{-1})^d$ (where $d$ is the depth of $v$). Note that $\alpha_v \geq 5$, since $d \leq \log(K^{100})$.
\end{itemize}
For these parameters our goal is to compute $\Delta_{v, -K}, \dots, \Delta_{v, K}$ in the sense of \cref{def:problem}.

\paragraph{Correctness}
We start with the correctness proof.

\begin{lemma}[Correctness of \cref{alg:main}] \label{lem:main-correctness}
Let $X, Y$ be strings. Given any node $v$ in the partition tree, \cref{alg:main} correctly solves the Tree Distance Problem.
\end{lemma}
\begin{proof}
The analysis of \crefrange{alg:main:line:iter-children}{alg:main:line:return} (that is, combining the recursive computations) is precisely as in \cref{lem:ako-correctness}. We therefore omit the details an assume that these steps succeed. In this proof we show that \crefrange{alg:main:line:test-short-condition}{alg:main:line:random-rule} are correct as well. (We postpone the error analysis to the end of the proof.) There are three possible cases:
\begin{itemize}
\itemdesc{The Strings are Short:} First assume that $|X_v| \leq 100 K^2$ or that $v$ is a leaf node, in which case the condition in \cref{alg:main:line:test-short-condition} triggers. The algorithm computes and returns a multiplicative $2$\=/approximation $\Delta_{v, s}$ of $\ED(X_v, Y_{v, s})$ for all shifts~$s$ using \cref{thm:ed-shifts}. We need to justify that $2 \leq \alpha_v$ so that $\Delta_{v, s}$ is a valid approximation in the sense of \Cref{eq:delta-approx}. Indeed, using the parameter setting in the previous paragraph we have~$\alpha_v \geq 5$.
\end{itemize}
If the algorithm does not terminate in this first case, we may assume that $|X_v| \geq 100 K^2$. The algorithm continues running and applies the Matching Test (\cref{lem:alignment-test}) to $X_v, Y_v$ and the $4K$-Periodicity Test (\cref{lem:periodicity-test}) to $Y_v$, both with rate parameter~$r = 3r_v$ and error parameter~$\delta = 0.01 \cdot K^{-100}$. We again postpone the error analysis and assume that both tests returned a correct answer. We continue analyzing the remaining two cases:
\begin{itemize}[itemsep=\smallskipamount]
\itemdesc{The Strings are Periodic:} Assume that the Matching Test reports \Close[s^*] and that the $4K$-Periodicity Test reports \Close[P], where $P$ is a primitive string with length \makebox{$|P| \leq 4K$}. The algorithm reaches \cref{alg:main:line:periodicity-rule} and returns $\Delta_{v, s} = 2 \cdot \min_{j \in \Int} \big|\,s - s^* - j|P| \,\big|$. We argue that this approximation is valid using \cref{lem:periodicity-rule} and by applying the triangle inequality three times. To this end we define $Z = P^* \range{0}{|Y_v|}$ (that is, $Z$ is equal to $Y_v$ after ``correcting'' the periodicity errors) and $Z_s = Z \range{K + s}{|Z| - K + s}$. By \cref{lem:periodicity-rule} we have that
\begin{equation*}
    \ED(Z_{s^*}, Z_s) = 2 \cdot \min_{j \in \Int} \big|\, s - s^* - j|P| \,\big| = \Delta_{v, s},
\end{equation*}
for all shifts~$s$. Here we use the assumption that $|X_v| \geq 100K^2 \geq |P|^2$ and its consequence $|Z| = |Y_v| \geq |P|^2 + 2K$ to satisfy the precondition of \cref{lem:periodicity-rule}. Using that the Periodicity Test reported \Close[P], we infer that \makebox{$\HD(Y_{v, s}, Z_s) \leq 1/(3r_v)$} for all shifts~$s$, and using that the Matching Test reported \Close[s^*] we obtain $\HD(X_v, Y_{v, s^*}) \leq 1/(3r_v)$. By applying the triangle inequality three times we conclude that
\begin{align*}
    \Delta_{v, s} &= \ED(Z_{s^*}, Z_s) \\
    &\leq \ED(Z_{s^*}, Y_{v, s^*}) + \ED(Y_{v, s^*}, X_v) + \ED(X_v, Y_{v, s}) + \ED(Y_{v, s}, Z_s) \\
    &\leq \ED(X_v, Y_{v, s}) + 1/r_v,
\end{align*}
and similarly $\Delta_{v, s} \geq \ED(X_v, Y_{v, s}) - 1/r_v$. It follows that $\Delta_{v, s}$ is an additive $1/r_v$-approximation of $\ED(X_v, Y_{v, s})$, as required in \Cref{eq:delta-approx}. (Here, we do not suffer any multiplicative error.)

\itemdesc{The Strings are Random-Like:} Finally assume that the Matching Test reports \Close[s^*], but the $4K$-Periodicity Test reports \Far. In this case the algorithm reaches \cref{alg:main:line:random-rule} and continues with the recursive computation (in \cref{alg:main:line:iter-children}). However, if at any level in the computation subtree rooted at $v$ there are more than $20K$ active nodes for which the Periodicity Test (in \cref{alg:main:line:periodicity-test}) reports \Far, then the recursive computation is interrupted and we return $\Delta_{v, s} = 2 \cdot |s - s^*|$. We have already argued that the unrestricted recursive computation is correct, but it remains to justify why interrupting the computation makes sense.

So suppose that the recursive computation is interrupted, i.e., assume that there are more than $20K$ descendants $w$ of $v$ at some level for which the Periodicity Test reported~\Far. Assuming that all Periodicity Tests computed correct outputs, we conclude that for all these descendants $w$ the strings $Y_v$ are not $4K$-periodic. From \cref{lem:intermediate-efficient} we learn that necessarily~$\bp_{4K}(Y_v) > 10K$. Hence \cref{lem:random-rule} applies and yields that
\begin{equation*}
    \ED(Y_{v, s^*}, Y_{v, s}) = 2 \cdot |s - s^*| = \Delta_{v, s}.
\end{equation*}
Using again the triangle inequality and the assumption that $\ED(Y_{v, s^*}, X_v) \leq 1/r_v$ (by the Matching Test), we derive that
\begin{align*}
    \Delta_{v, s} &= \ED(Y_{v, s^*}, Y_{v, s}) \\
    &\leq \ED(Y_{v, s^*}, X_v) + \ED(X_v, Y_{v, s}) \\
    &\leq \ED(X_v, Y_{v, s}) + 1/r_v.
\end{align*}
The lower bound can be proved similarly and therefore $\Delta_{v, s}$ is an additive $1/r_v$-approximation of $\ED(X_v, Y_{v, s})$.
\end{itemize}
We finally analyze the error probability of \cref{alg:main}. There are three sources of randomness in the algorithm: The Matching and Periodicity Tests in \cref{alg:main:line:alignment-test,alg:main:line:periodicity-test} and the application of the Precision Sampling Lemma. For each node, therefore have three error events: With probability at most $2\delta = 0.02 \cdot K^{-100}$ one of the property tests fails. We apply the Precision Sampling Lemma with $\delta = 0.01 \cdot K^{-101}$ for $2K$ shifts in every node, hence the error probability is bounded by $0.02 \cdot K^{-100}$ as well. In summary: The error probability per node is $0.04 \cdot K^{100}$. Recall that there are at most $2K^{100}$ nodes in the partition tree, and thus the total error probability is bounded by $0.08 \leq 0.1$.
\end{proof}

\paragraph{Running Time}
This concludes the correctness part of the analysis and we continue bounding the running time of \cref{alg:main}. We proceed in two steps: First, we give an upper bound on the number of active nodes in the partition tree (see \cref{lem:alignment-many,lem:main-active-nodes}). Second, we bound the expected running time of a single execution of \cref{alg:main} (ignoring the cost of recursive calls). The expected running time is bounded by their product.

Recall that a node $v$ is \emph{matched} if there is some $s \in \set{-K, \dots, K}$ such that $X_v = Y_{v, s}$. Moreover, we say that $v$ is \emph{active} if the recursive computation of \cref{alg:main} reaches $v$.

\begin{lemma}[Number of Unmatched Nodes] \label{lem:alignment-many}
Assume that $\ED(X, Y) \leq K$. If the partition tree has depth $D$, then there are at most $K D$ nodes which are not matched.
\end{lemma}
\begin{proof}
Focus on any level in the partition tree and let $0 = i_0 < \dots < i_w = n$ denote the partition induced by that level, i.e., let $\range{i_\ell}{i_{\ell+1}} = I_v$ where $v$ is the $\ell$-th node in the level (from left to right). Let $A$ be an optimal alignment between $X$ and $Y$, then:
\begin{equation*}
    \ED(X, Y) = \sum_{\ell=0}^{w-1} \ED(X \range{i_\ell}{i_{\ell+1}}, Y \range{A(i_\ell)}{A(i_{\ell+1})}).
\end{equation*}
Since we assumed that $\ED(X, Y) \leq K$, there can be at most $K$ nonzero terms in the sum. For any zero term we have that $X \range{i_\ell}{i_{\ell+1}} = Y \range{A(i_\ell)}{A(i_{\ell+1})}$ and therefore the $\ell$-th node in the current level is matched with shift $A(i_\ell) - i_\ell$. By \cref{prop:alignment-stretch} we have that $|A(i_\ell) - i_\ell| \leq \ED(X, Y) \leq K$. This completes the proof.
\end{proof}

\begin{lemma}[Number of Active Nodes] \label{lem:main-active-nodes}
Assume that $\ED(X, Y) \leq K$. If the partition tree has depth $D$, then there are at most $\Order((K D B)^2)$ active nodes, with probability~$0.98$.
\end{lemma}
\begin{proof}
Recall that (unconditionally) there are at most $2K^{100}$ nodes in the partition tree. Hence, by a union bound, all Matching Tests in \cref{alg:main:line:alignment-test} succeed with probability at least $1 - 0.02 = 0.98$. We will condition on this event throughout the proof. We distinguish between three kinds of nodes $v$:
\begin{enumerate}
\item $v$ itself and all of $v$'s ancestors are not matched,
\item $v$ itself is matched, but all of $v$'s ancestors are not matched,
\item some ancestor of $v$ (and therefore also $v$ itself) is matched.
\end{enumerate}

By the previous lemma we know that there are at most $K D$ nodes which are not matched. It follows that there are at most $K D$ nodes of the first kind.

It is also easy to bound the number of nodes $v$ of the second kind: Observe that $v$'s parent is a node of the first kind. Hence, there can be at most $KD \cdot B$ nodes of the second kind.

Finally, we bound the number of nodes of the third kind. Any such node $v$ has a unique ancestor $w$ of the second kind. There are two cases for~$w$: Either the condition in \cref{alg:main:line:periodicity-test-condition} succeeds and the algorithm directly solves $w$. This is a contradiction since we assumed that $v$ (a descendant of $w$) is active. Or this condition fails, and the algorithm continues branching with the exception that if in the subtree below $w$ there are more than $20K$ active nodes per level for which the Periodicity Test in \cref{alg:main:line:periodicity-test} reports \Far, then we interrupt the recursive computation. We claim that consequently in the subtree below~$w$ (consisting only of nodes of the third kind), there are at most~$20 K B$ active nodes per level. Indeed, suppose there were more than~$20 K B$ active nodes on some level. Then consider their parent nodes; there must be more than~$20K$ parents. For each such parent $u$ the Matching Test reported \Close{} (since~$u$ is a matched node, and we assumed that all Matching Tests succeed) and the Periodicity Test reported \Far{} (since otherwise the condition in \cref{alg:main:line:periodicity-test} triggers and solves $u$ directly, but we assumed that $u$ is the parent of some other active node). Note that we have witnessed more than $20K$ nodes on one level below $w$ for which the Periodicity Test reported \Far. This is a contradiction.

In total the number of active nodes below $w$ is bounded by $20 K B \cdot D$. Recall that there are at most $KD \cdot B$ nodes $w$ of the second kind, hence the total number of active nodes of the third kind is $20 (K D B)^2$. Summing over all three kinds, we obtain the claimed bound.
\end{proof}

We are ready to bound the total running time. In the following lemma we prove that the algorithm is efficient \emph{assuming that the edit distance between $X$ and $Y$ is small}. This assumption can be justified by applying \cref{alg:main} with the following modification: We run \cref{alg:main} with a time budget and interrupt the computation as soon as the budget is depleted. In this case we can immediately infer that the edit distance between $X$ and $Y$ must be large.

\begin{lemma}[Running Time of \cref{alg:main}] \label{lem:main-time}
Let $X, Y$ be strings with $\ED(X, Y) \leq K$. Then \cref{alg:main} runs in time $n / K \cdot (\log K)^{\Order(\log_B(K))} + \widetilde\Order(K^4 B^2)$, with constant probability $0.9$.
\end{lemma}
\begin{proof}
We first bound the expected running time of a single execution of \cref{alg:main} where we ignore the cost of recursive calls. Let $D$ denote the depth of the partition tree. We proceed in the order of the pseudocode:
\begin{itemize}
\item \cref{alg:main:line:test-short-condition,alg:main:line:test-short}: If the test in \cref{alg:main:line:test-short-condition} succeeds, then \cref{alg:main:line:test-short} takes time~$\Order(|X_v| + K^2)$ by \cref{thm:ed-shifts}, where $|X_v| \leq 100K^2$ or $|X_v| \leq \Order(n / K^{100})$. The total time of this step is therefore bounded by~$\Order(K^2 + n / K^{100})$.
\item \cref{alg:main:line:alignment-test,alg:main:line:periodicity-test}: Running the Matching and Periodicity Tests (\cref{lem:alignment-test,lem:periodicity-test}) takes time $\Order(r_v |X_v| \log K + K \log |X_v|)$. We want to replace the $\log|X|$ by $\log K$ here, so assume that the second summand dominates, i.e., \makebox{$r_v |X_v| \log K \leq K \log |X_v|$}. At any node $v$ the rate $r_v$ is always at least $1000/K$ (since at the root the rate is exactly $1000/K$ and below the root the rate never decreases), hence $|X_v| / \log |X_v| \leq \poly(K)$. It follows that we can bound the total time of this step indeed by $\Order(r_v |X_v| \log K + K \log K)$.
\item \crefrange{alg:main:line:alignment-test-condition}{alg:main:line:random-rule}: Here we merely produce the output according to some fixed rules. The time of this step is bounded by $\Order(K)$.
\item \crefrange{alg:main:line:iter-children}{alg:main:line:return}: This step takes time $\Order(K B)$ and the analysis is exactly as in \cref{lem:ako-time}.
\end{itemize}
In total, the time of a single execution is bounded by $\Order(r_v |X_v| \log K + K^2 + n/K^{100})$. We will simplify this term by plugging in the (expected) rate~$r_v$ for any node $v$ (in a way similar to \cref{lem:ako-time}).

Recall that~$r_v = 1000 \cdot (K \cdot u_{v_1} \dots u_{v_d})^{-1}$ where $v_0, v_1, \dots, v_d = v$ is the root-to-node path leading to $v$ and each $u_i$ is sampled from \makebox{$\mathcal D(\varepsilon = (200 \log K)^{-1}, \delta = 0.01 \cdot K^{-101})$}, independently. Using \cref{lem:precision-sampling} there exist events $E_w$ happening each with probability $1 - 1/N$ such that
\begin{equation*}
    \Ex(1 / u_w \mid E_w) \leq \widetilde{\Order}(\varepsilon^{-2} \log(\delta^{-1}) \log N) \leq \polylog(K).
\end{equation*}
In the last step we set $N = 100 K^{100}$. Taking a union bound over all active nodes $w$ (there are at most $2K^{100}$ many), the event $E = \bigwedge_w E_w$ happens with probability at least $0.98$ and we will condition on $E$ from now on. Under this condition we have:
\begin{equation*}
    \Ex(r_v \mid E) = \frac{1000}K \prod_{i=1}^d \Ex(1 / u_{v_i} \mid E_{v_i}) \leq \frac{(\log K)^{\Order(d)}}K \leq \frac{(\log K)^{\Order(\log_B(K))}}K.
\end{equation*}

Finally, we can bound the total expected running time (conditioned on $E$) as follows, summing over all active nodes $v$:
\begin{equation*}
    \sum_v \Order\!\left(|X_v| \cdot \frac{(\log K)^{\Order(\log_B(K))}}k + K^2 + \frac{n}{K^{100}}\right)\!.
\end{equation*}
Using that $\sum_w |X_w| = n$ whenever $w$ ranges over all nodes on a fixed level in the partition tree, and thus $\sum_v |X_v| \leq n \cdot D$ where $v$ ranges over all nodes, we can bound the first term in the sum by $n / K \cdot (\log K)^{\Order(\log_B(K))}$. The second term can be bounded by $K^2$ times the number of active nodes. By \cref{lem:main-active-nodes} this becomes $\Order(K^4 D^2 B^2) = \widetilde\Order(K^4 B^2)$. By the same argument the third term becomes at most $n/K^{90}$ and is therefore negligible.

We conditioned on two events: The event $E$ and the event that the number of active nodes is bounded by $\Order(K^2 D^2 B^2)$ (\cref{lem:main-active-nodes}). Both happen with probability at least $0.98$, thus the total success probability is $0.96 \geq 0.9$.
\end{proof}

\subsection{Main Theorem} \label{sec:main:sec:main-theorem}
We finally recap and formally prove our main theorem and its two corollaries.

\thmmain*
\begin{proof}
For now we keep $K$ as a parameter and will later set $K$ in terms of $k$. We run \cref{alg:main} to compute an approximation~$\Delta = \Delta_{r, 0}$ where $r$ is the root node in the partition tree. By the correctness of \cref{alg:main} (\cref{lem:main-correctness}) we have that $0.1 \ED^{\leq K}(X, Y) - 0.001 K \leq \Delta \leq 10 \TD^{\leq K}(X, Y) + 0.001 K$, and using the equivalence of edit distance and tree distance (\cref{lem:equivalence-ed-td}) we conclude that $\Delta \leq 20 B D \cdot \ED^{\leq K}(X, Y) + 0.001 K$, where~$D \leq \log_B((K)^{100})$ is the depth of the partition tree. It follows that we can distinguish whether the edit distance $\ED(X, Y)$ is at most $K / (1000 BD)$ or at least $K$. Indeed:
\begin{itemize}
\item If $\ED(X, Y) \leq K / (1000 BD)$, then $\Delta \leq 0.02 K + 0.001K = 0.021 K$.
\item If $\ED(X, Y) \geq K$, then $\Delta \geq 0.1 K - 0.001 K = 0.099 K$. 
\end{itemize}

To bound the running time, we run the previous algorithm with time budget $n / K \cdot (\log K)^{\Order(\log_B(K))} + \widetilde\Order(K^4 B^2)$ (with the same constants as in \cref{lem:main-time}). If the algorithm exceeds the time budget, then we interrupt the computation and report that $\ED(X, Y) \geq K$. This is indeed valid, since \cref{lem:main-time} certifies that $\ED(X, Y) > K$ in this case.

To obtain the claimed statement, we have to pick $K$. We set $K = \Theta(k \log_B(k) \cdot B)$, where the constant is picked in such a way that $K / (1000 B D) \geq K / (1000 B \cdot \log_B((K)^{100})) \geq k$. Then the algorithm distinguishes edit distances $k$ versus $K = \Theta(k \log_B(k) \cdot B)$.
\end{proof}

\coroptimaltime*
\begin{proof}
Simply plugging $B = 2^{\sqrt{\log k}}$ into \cref{thm:main} leads to the correct gap, but we suffer a factor $k^{\order(1)}$ in the running time. For that reason, let $\bar k$ be a parameter to be specified later and apply \cref{thm:main} with parameter~$\bar k$ and $B = 2^{\sqrt{\log \bar k}}$. In that way we can distinguish the gap $\bar k$ versus \makebox{$\bar k \cdot 2^{\Theta(\sqrt{\log \bar k})}$} in time $\Order(n / \bar k \cdot (\log \bar k)^{\Order(\log_B(\bar k))} + \bar k^{4+\order(1)})$. This term can be written as \makebox{$\Order(n / \bar k \cdot 2^{\alpha(\bar k)} + \bar k^{4+\order(1)})$}, where $\alpha(\bar k) = \Order(\sqrt{\log \bar k} \log\log \bar k)$.

Finally, set $\bar k = k \cdot 2^{2\alpha(k)}$. For $k$ at least a sufficiently large constant we have that $2^{\alpha(\bar k)} \leq 2^{2\alpha(k)}$, and therefore the running time becomes $\Order(n/k + k^{4+\order(1)})$ as claimed. (For small constant $k$, we can exactly compute the $k$-capped edit distance in linear time $\Order(n)$ using the Landau-Vishkin algorithm~\cite{LandauV88}.) The algorithm distinguishes the gap $\bar k$ versus $\bar k \cdot 2^{\Order(\sqrt{\log k})} = k \cdot 2^{\widetilde\Theta(\sqrt{\log k})}$. Since $k \leq \bar k$, this is sufficient to prove the claim.
\end{proof}

\corpolyloggap*
\begin{proof}
Let $c$ be the constant so that the time bound in \cref{thm:main} becomes $n / k \cdot (\log k)^{c \log_B(k)} + \widetilde\Order(k^4 \poly(B))$. We apply \cref{thm:main} with parameter $B = (\log k)^{c / \varepsilon}$. Then the gap is indeed $k$ versus~$k \cdot (\log k)^{\Order(1/\varepsilon)}$ as claimed. Since $(\log k)^{c \log_B(k)} = k^\varepsilon$ the running time bound becomes~$\Order(n / k^{1-\varepsilon} + k^{4+\order(1)})$.
\end{proof}

\newpage
\appendix
\section{Equivalence of Edit Distance and Tree Distance} \label{sec:ed-td}
In this section we prove that the tree distance closely approximates the edit distance. The proof is an easy generalization of~\cite[Theorem~3.3]{AndoniKO10}. At the end of the section, we also provide a proof of \cref{lem:capped-td-equivalence}.
\lemequivedtd*

For the proof we continue using the notation $X_v$ and $Y_{v, s}$ as in the rest of the paper. That is, for a node $v$ with $I_v = \range ij$ we set $X_v = X \range ij$ and $Y_{v, s} = Y \range{i + s}{j + s}$.

\paragraph{The Lower Bound}
We prove the more general statement $\ED(X_v, Y_{v, s}) \leq \TD_{v, s}(X, Y)$. The proof is by induction on the depth of $T$. If $v$ is a leaf we have $\ED(X_v, Y_{v, s}) = \TD_{v, s}(X, Y)$ by definition. So focus on the case where $v$ is an internal node with children $v_0, \dots, v_{b-1}$. In this case we have:
\begin{align}
    \ED(X_v, Y_{v, s}) &\leq \sum_{\ell=0}^{b-1} \ED(X_{v_\ell}, Y_{v_\ell, s}) \label{lem:equiv-ed-td:eq:1} \\
    &\leq \sum_{\ell=0}^{b-1} \min_{s' \in \Int} \ED(X_{v_\ell}, Y_{v_\ell, s'}) + \ED(Y_{v_\ell, s'}, Y_{v_\ell, s}) \label{lem:equiv-ed-td:eq:2} \\
    &\leq \sum_{\ell=0}^{b-1} \min_{s' \in \Int} \ED(X_{v_\ell}, Y_{v_\ell, s'}) + 2 \cdot |s - s'| \label{lem:equiv-ed-td:eq:3} \\
    &\leq \sum_{\ell=0}^{b-1} \min_{s' \in \Int} \TD_{v_\ell, s'}(X, Y) + 2 \cdot |s - s'| \label{lem:equiv-ed-td:eq:4} \\
    &= \TD_{v, s}(X, Y). \label{lem:equiv-ed-td:eq:5}
\end{align}
Here we used \eqref{lem:equiv-ed-td:eq:1} the facts that $X_v = \bigcirc_\ell X_{v_\ell}$ and $Y_{v, s} = \bigcirc_\ell Y_{v_\ell, s}$, \eqref{lem:equiv-ed-td:eq:2} the triangle inequality, \eqref{lem:equiv-ed-td:eq:3} the observation that $\ED(Y_{w, s}, Y_{w, s'}) \leq 2 \cdot |s - s'|$ for all shifts $s, s' \in \Int$ (by deleting~$|s - s'|$ in the beginning and inserting~$|s - s'|$ characters at the end of the strings), \eqref{lem:equiv-ed-td:eq:4} the induction hypothesis, and finally \eqref{lem:equiv-ed-td:eq:5} the definition of the tree distance.

\paragraph{The Upper Bound}
Because of some technical complications we do not prove the upper bound by induction. Instead, we will unfold the definition of tree distance and prove the upper bound ``globally''. Specifically, it is easy to prove (by induction on the depth of the computation tree) that
\begin{equation*}
    \TD(X, Y) = \min_{\substack{s \in \Int^T\\s_r = 0}} \;\Bigg(\sum_{\substack{\text{internal}\\\text{nodes $v$}}} \;\sum_{\substack{\text{children $w$}\\\text{of $v$}}} 2 \cdot |s_v - s_w|\Bigg) + \sum_{\text{leaves $v$}} \ED(X_v, Y_{v, s_v}),
\end{equation*}
where we write $r$ for the root node. For clarity: Here, we minimize over values $s_v \in \Int$ for all nodes $v$ in the computation tree---except for the root node $r$ for which we fix $s_r = 0$.

To prove an upper bound on this expression, we first specify the values $s_v$. For this purpose, we let $A$ denote an optimal alignment between $X$ and $Y$. Recall that the computation tree determines an interval $I_v$ for each node $v$; we write $I_v = \range{i_v}{j_v}$. We can now pick the value $s_v = A(i_v) - i_v$. For convenience we also define $Y_v' = Y \range{A(i_v)}{A(j_v)}$. By the way we defined optimal alignments $A$, we have that $\ED(X, Y) = \sum_v \ED(X_v, Y_v')$, where the sum is over all nodes~$v$ in a cut through the partition tree (e.g., all leaves or all nodes at one specific level).

\begin{claim*}
$\ED(X_v, Y_{v, s_v}) \leq 2 \ED(X_v, Y_v')$.
\end{claim*}
\begin{proof}
By the triangle inequality, we can bound $\ED(X_v, Y_{v, s_v})$ by the sum of $\ED(X_v, Y_v')$ and $\ED(Y_v', Y_{v, s_v})$. But observe that both strings $Y_v'$ and $Y_{v, s_v}$ are substrings of $Y$ starting at the same position $A(i_v)$. Hence, one is a prefix of the other string and we can bound their edit distance by their length distance, $\ED(Y_v', Y_{v, s_v}) \leq |\, |Y_v'| - |Y_{v, s_v}| \,| = |\, |Y_v'| - |X_v| \,| \leq \ED(X_v, Y_v')$.
\end{proof}

\begin{claim*}
Let $w$ be a child of $v$. Then $|s_v - s_w| \leq \ED(X_v, Y_v')$.
\end{claim*}
\begin{proof}
Recall that $s_v = A(i_v) - i_v$ and $s_w = A(i_w) - i_w$. Using that the edit distance of two strings is at least their difference in length, we have that $|s_v - s_w| \leq \ED(X \range{i_v}{i_w}, Y \range{A(i_v)}{A(i_w)})$. Since $w$ is a child of $v$ we have that $i_w \leq j_v$. Thus we obtain $|s_v - s_w| \leq \ED(X \range{i_v}{i_w}, Y \range{A(i_v)}{A(i_w)}) = \ED(X_v, Y_v')$, using again that $A$ is an optimal alignment.
\end{proof}

Using the two claims, we can finally give an upper bound on the tree distance. In the following calculation we will use \eqref{lem:equiv-ed-td:eq:6} the tree distance characterization from before (plugging in our values $s_v$), \eqref{lem:equiv-ed-td:eq:7} the two claims, \eqref{lem:equiv-ed-td:eq:8} the assumption that each node has at most~$B$ children and the (obvious) fact that each node is either a leaf or internal, and finally \eqref{lem:equiv-ed-td:eq:9} that $\sum_v \ED(X_v, Y_v') = \ED(X, Y)$ whenever the sum is over all nodes $v$ at one specific level of the computation tree, hence $\sum_v \ED(X_v, Y_v') = D \cdot \ED(X, Y)$ if the sum is over all nodes~$v$: 
\begin{align}
    \TD(X, Y) &\leq \Bigg(\sum_{\substack{\text{internal}\\\text{nodes $v$}}} \;\sum_{\substack{\text{children $w$}\\\text{of $v$}}} 2 \cdot |s_v - s_w|\Bigg) + \sum_{\text{leaves $v$}} \ED(X_v, Y_{v, s_v}) \label{lem:equiv-ed-td:eq:6} \\
    &\leq \Bigg(\sum_{\substack{\text{internal}\\\text{nodes $v$}}} \;\sum_{\substack{\text{children $w$}\\\text{of $v$}}} 2 \cdot \ED(X_v, Y_v')\Bigg) + \sum_{\text{leaves $v$}} 2 \cdot \ED(X_v, Y_v') \label{lem:equiv-ed-td:eq:7} \\
    &\leq 2B \sum_{\substack{\text{all}\\\text{nodes $v$}}} \ED(X_v, Y_v') \label{lem:equiv-ed-td:eq:8} \\
    &\leq 2BD \cdot \ED(X, Y). \label{lem:equiv-ed-td:eq:9}
\end{align}
This completes the proof of \cref{lem:equivalence-ed-td}. It remains to prove \cref{lem:capped-td-equivalence}.

\lemcappedtdequivalence*
\begin{proof}
It is easy to prove by induction that $\TD_{v, s}^{\leq K}(X, Y) \geq \min(\TD_{v, s}(X, Y), K)$ for all nodes $v$ and all shifts $s$, however, the other direction is not necessarily true for all nodes.


We prove that nevertheless $\TD^{\leq K}(X, Y) \leq \min(\TD(X, Y), K)$. We may assume that $\TD(X, Y) < K$ as otherwise the statement is clear. Let $s_v$ denote the optimal shifts picked at all nodes $v$ in the tree distance definition~\eqref{eq:tree-distance}. That is, we have $s_r = 0$ for the root node~$r$ and for all nodes $v$ with children~$v_1, \dots, v_B$ the following equation is satisfied:
\begin{equation*}
    \TD_{v, s_v}(X, Y) = \sum_{i \in \rangezero B} (\TD_{v_i, s_{v_i}}(X, Y) + 2 \cdot |s_v - s_{v_i}|).
\end{equation*}
By induction one can easily verify that $|s_v| \leq \TD(X, Y) < K$ for all nodes $v$. But this implies that the same shifts can be picked in the $K$-capped tree distance definition~\eqref{eq:capped-td} to certify that $\TD^{\leq K}_{r, 0}(X, Y) \leq \TD_{r, 0}(X, Y)$, and thus $\TD^{\leq K}(X, Y) \leq \TD(X, Y)$. 
\end{proof}
\section{Precision Sampling Lemma}\label{sec:psl}
In this section, we give a proof of the Precision Sampling Lemma. Although originally used and stated by Andoni, Krauthgamer and Onak in~\cite{AndoniKO10}, our formulation is essentially taken from their follow-up paper~\cite{AndoniKO11}. We give a simpler proof using the exponential distribution, as suggested in a later paper by Andoni~\cite{Andoni17} (indeed, Andoni states a simplified version of the precision Sampling Lemma in~\cite{Andoni17} to show how exponential random variables simplify the proof). 

Recall that the exponential distribution with rate $\lambda > 0$ is a continuous distribution over the positive reals with probability density function $f(x) = \lambda e^{-\lambda x}$. If a random variable $u$ has this distribution, we write $u \sim \Exp(\lambda)$. We will use the following facts:

\begin{fact}[Properties of the Exponential Distribution]\label{fact:exponential-distribution}
    The following holds:
    \begin{itemize}
    \itemdesc{Scaling:} Let $u \sim \Exp(\lambda)$ and $\alpha > 0$, then $u/\alpha \sim \Exp(\alpha \cdot \lambda)$.
    \itemdesc{Min-Stability:} If $u_1,\dots,u_n$ are independent and $u_i \sim \Exp(\lambda_i)$, then $\min\{u_1,\dots,u_n\}$ is distributed as $\Exp(\sum_i \lambda_i)$.
    \end{itemize}
\end{fact}

Recall that we say that $\widetilde x$ is an $(a,b)$-approximation of $x$ if $x/a - b \leq \widetilde x \leq ax + b$. We will use the following composition property which follows immediately:

\begin{proposition}[Composition of Approximations]\label{prop:approx-composition}
Let $\widetilde{x}$ be an $(a,b)$-approximation of $x$ and $\widetilde{y}$ be an $(a',b')$-approximation of $\widetilde{x}$. Then, $\widetilde{y}$ is an $(a \cdot a', a' \cdot b + b')$-approximation of $x$.
\end{proposition}

\lempsl*

\begin{algorithm}[t]
\caption{Precision Sampling} \label{alg:psl}
\begin{algorithmic}[1]
    \medskip
    \State Sample $u_{i,j} \sim \Exp(1)$ for $i \in \rangezero{n}, j \in \rangezero{\lambda}$ independently \label{alg:psl:line:precisions-1}
    \State Let $u_i = \min_j u_{i,j}$ for each $i \in \rangezero{n}$ \label{alg:psl:line:precisions-2}
    \State Obtain $(a, b \cdot u_i)$-approximations $\widetilde{A}_i$ of $A_i$ for all $i \in \rangezero{n}$ \label{alg:psl:line:approximations}
    \State Let $\widetilde{M}_j = \max_{i \in \rangezero{n}} \widetilde{A}_i / u_{i,j}$ for all $j \in \rangezero{\lambda}$ \label{alg:psl:line:max}
    \State \Return $\ln 2 \cdot \median_{j \in \rangezero{\lambda}} \widetilde{M}_j$ \label{alg:psl:line:return}
\end{algorithmic}
\end{algorithm}

\begin{proof}
We start by explaining the recovery algorithm, whose pseudocode is in~\cref{alg:psl}. \Cref{alg:psl:line:precisions-1,alg:psl:line:precisions-2} choose the \emph{precisions} $u_i$ which determine the additive error required from the approximations $\widetilde{A}_i$ obtained in \cref{alg:psl:line:approximations} for each $i \in \rangezero{n}$. To understand \cref{alg:psl:line:max}, observe that since each $\widetilde{A}_i$ is an $(a, b \cdot u_i)$-approximation of $A_i$ and since $u_i \leq u_{i,j}$, it follows that $\widetilde{M}_j$ is an $(a,b)$-approximation of 
\begin{equation*}
    M_j = \max_{i \in \rangezero n} \frac{A_i}{u_{i, j}}.
\end{equation*}
In what follows, we will argue that the median of $M_1,\dots,M_{\lambda}$ scaled by $\ln 2$ is an $(1 + \varepsilon, 0)$-approximation of $\sum_i A_i$ with probability at least $1 - \delta$. Then, by~\cref{prop:approx-composition} the scaled median of the \raisebox{0pt}[0pt][0pt]{$\widetilde{M}_j$}'s returned in~\cref{alg:psl:line:return} gives the desired $((1+\varepsilon) \cdot a, b)$-approximation of~$\sum_i A_i$.

Using the scaling and min-stability properties of \cref{fact:exponential-distribution}, we have that $M_j$ is the inverse of an exponentially distributed random variable with rate $\sum_i A_i$. Hence, since $\ln 2$ is the median of $\Exp(1)$ and it has a continuous probability density function, we have that 
\begin{align*}
    \Pr(M_j \cdot \ln 2 < (1 - \varepsilon){\textstyle \sum_i A_i}) 
        &= \Pr_{u \sim \Exp(1)}(\tfrac{1}{u} < (1 - \varepsilon)/\ln 2)\\
        &< \Pr_{u \sim \Exp(1)}(u > (1 + \varepsilon/2) \ln 2) = \frac{1}{2} - \Theta(\varepsilon).
\end{align*}
Similarly, it holds that $\Pr(M_j \cdot \ln 2 > (1 + \varepsilon)\sum_i A_i) < 1/2 - \Theta(\varepsilon)$.

Let $X^-$ denote the number of $j$'s with $M_j \cdot \ln 2 < (1 - \varepsilon)\sum_i A_i$ and $X^+$ the number of $j$'s with $M_j \cdot \ln 2 > (1 + \varepsilon) \sum_i A_i$. The scaled median of the $M_j$'s will be an $(1 + \varepsilon, 0)$-approximation to $\sum_i A_i$ if $X^- < \lambda/2$ and $X^+ < \lambda/2$. We can bound the probability of these events with a Chernoff bound:
\begin{gather*}
    \Pr(X^- > \lambda/2) = \Pr(X^- > \Ex[\,X^-\,] + \Theta(\lambda \varepsilon)) \leq \exp(-\Omega(\lambda \varepsilon^2)),\\
    \Pr(X^+ > \lambda/2) = \Pr(X^+ > \Ex[\,X^+\,] + \Theta(\lambda \varepsilon)) \leq \exp(-\Omega(\lambda \varepsilon^2)).
\end{gather*}
Putting things together, we have that setting $\lambda = \Order(\varepsilon^{-2} \log(\delta^{-1}))$, the value returned in \cref{alg:psl:line:return} of \cref{alg:psl} is an $((1+\varepsilon) \cdot a, b)$-approximation to $\sum_i A_i$ with probability at least~$1 - \delta$. 

Note that our description so far does not quite adhere to the accuracy property of the Lemma (more precisely, the recovery algorithm as stated needs to know the samples~$u_{i,j}$ which were used to define the precisions $u_i$). To fix this technicality, we note that alternatively to \cref{alg:psl:line:precisions-1,alg:psl:line:precisions-2} of \cref{alg:psl} we could first sample $u_i \sim \Exp(\lambda)$ directly and then for \cref{alg:psl:line:max} sample~$u_{i,j} \sim \Exp(1)$ conditioned on that $\min_j u_{i,j} = u_i$. Due to the min-stability of the exponential distribution (\cref{fact:exponential-distribution}) this gives the same distribution (but the above is easier to implement). In particular, in this way the distribution $\mathcal{D}$ in the statement is just  $\Exp(\lambda)$.

\medskip
Now we prove the efficiency property. Fix an arbitrary $N \geq 1$ and let $u \sim \mathcal{D}$. Let $E$ denote the event that $u > 1 / (\lambda N)$. Then $\Pr(E) = e^{-1/N} \geq 1/3$, and
\begin{align*}
    \Ex[\,1 / u \mid E\,] 
        &= \frac{1}{\Pr(E)} \int_{u=1 / (\lambda N)}^\infty \frac{1}{u} \cdot \lambda e^{-\lambda u} du \\
        &\leq 3 \int_{u=1/(\lambda N)}^1 \frac{1}{u} \cdot \lambda e^{-\lambda u} du + 3 \int_{u = 1}^{\infty} \frac{1}{u} \cdot \lambda e^{-\lambda u} du \\
        &\leq 3 \lambda \int_{u= 1/(\lambda N)}^1 \frac{1}{u}\; du + O(1) \\
        &= \Order(\lambda \log(\lambda N)) \\
        &= \widetilde{O}(\lambda \log N).
\end{align*}
Plugging in our choice of $\lambda$, we obtain that $\Ex[\,1 / u \mid E\,] = \widetilde\Order(\varepsilon^{-2} \log(\delta^{-1}) \log N)$, as claimed. 

\medskip
Finally, we enforce that the distribution $\mathcal D$ is supported over $(0, 1]$, as the lemma states. We can simply transform a sample $u \sim \mathcal D$ into $u' \gets \min(u, 1)$. Since the samples $u$ only become smaller in this way, the estimates obtained \cref{alg:psl:line:approximations} of \cref{alg:psl} have smaller additive error and we therefore preserve the accuracy property. For the efficiency property, note that $\Ex[\,1 / u' \mid E\,] \leq \Ex[\,1 / u \mid E\,] + 1$, and therefore our bound on the conditional expectation remains valid asymptotically.
\end{proof}

\paragraph{Implementation in the Word RAM Model}
Throughout we implicitly assumed that we can sample from the continuous distribution $\Exp(\lambda)$, but in the word RAM model we obviously have to discretize. An easy way is to appropriately round the samples to multiples of~$1/\poly(n)$. The resulting distribution is a geometric distribution which we can easily sample from. Alternatively, one can prove the Precision Sampling Lemma directly in terms of a discrete probability distribution, see for instance~\cite{AndoniKO10}.
\section{The Range Minimum Problem} \label{sec:range-minimum}

In this section we give a simple algorithm to efficiently combine the recursive results from the children at every node of the computation tree.

\lemrangeminimum*
\begin{algorithm}[t]
\caption{} \label{alg:range-minimum}
\begin{algorithmic}[1]
\Input{Integers $A_{-K}, \dots, A_K$}
\Output{Integers $B_{-K}, \dots, B_K$ where $B_s = \min_{-K \leq s' \leq K} A_{s'} + 2  \cdot |s' - s|$}
\medskip
\State Let $B_{-K}^L \gets A_{-K}$ and $B_K^R \gets A_K$
\For{$s \gets -K+1, \dots, K$}
    \State $B_s^L \gets \min(B_{s-1}^L + 2, A_s)$
\EndFor
\For{$s \gets K-1, \dots, -K$}
    \State $B_s^R \gets \min(B_{s+1}^R - 2, A_s)$
\EndFor
\State\Return $B_s \gets \min(B_s^L, B_s^R)$ for all $s \in \set{-K, \dots, K}$
\end{algorithmic}
\end{algorithm}
\begin{proof}
We give the pseudocode in \cref{alg:range-minimum}. We claim that the algorithm correctly computes
\begin{gather*}
    B_s^L = \min_{-K \leq s' \leq s} A_{s'} - 2s' + 2s, \\
    B_s^R = \min_{s \leq s' \leq K} A_{s'} + 2s' - 2s,
\end{gather*}
for all $s$. This claim immediately implies that we correctly output $B_s = \min(B_s^L, B_s^R)$. We prove that $B_s^L$ is computed correctly; the other statement is symmetric. As the base case, the algorithm correctly assigns $B^L_{-K} = A_{-K}$. We may therefore inductively assume that~$B^L_{s-1}$ is assigned correctly, and have to show that $B^L_s \gets \min(B^L_{s-1} + 2, A_s)$ is a correct assignment. There are two cases: Either $B_s^L$ attains its minimum for $s' < s$, in which case~$B_s^L = B^L_{s-1} + 2$. Or $B_s^L$ attains the minimum for~$s' = s$, in which case~$B_s^L = A_s$. The total time is bounded by $\Order(K)$.
\end{proof}
\section{The Andoni-Krauthgamer-Onak Algorithm} \label{sec:ako}

\begin{algorithm}[t]
\caption{The Andoni-Krauthgamer-Onak algorithm} \label{alg:ako-full}
\begin{algorithmic}[1]
\Input{Strings $X, Y$, a node $v$ in the partition tree $T$ and a sampling rate $r_v > 0$}
\Output{$\Delta_{v, s}$ for all shifts $s \in \set{-K, \dots, K}$}
\medskip
\If{$v$ is a leaf} \label{alg:ako-full:line:test-1-condition}
    \State\Return $\Delta_{v, s} = \ED(X_v, Y_{v, s})$ for all $s \in \set{-K, \dots, K}$ \label{alg:ako-full:line:test-1}
\EndIf
\If{$|X_v| \leq 1/r_v$} \label{alg:ako-full:line:test-rate-condition}
    \State\Return $\Delta_{v, s} = 0$ for all $s \in \set{-K, \dots, K}$ \label{alg:ako-full:line:test-rate}
\EndIf
\ForEach{$i \in \rangezero B$} \label{alg:ako-full:line:iter-children}
    \State Let $v_i$ be the $i$-th child of $v$ and sample $u_{v_i} \sim \mathcal D((2 \log n)^{-1}, 0.01 \cdot (Kn)^{-1})$ \label{alg:ako-full:line:precision}
    \State Recursively compute $\Delta_{v_i, s}$ with rate $r_v / u_{v_i}$ for all $s \in \set{-K, \dots, K}$ \label{alg:ako-full:line:recursion}
    \State Compute $\widetilde A_{i, s} = \min_{s' \in \set{-K, \dots, K}} \Delta_{v_i, s'} + 2 \cdot |s - s'|$ using \cref{lem:range-minimum} \label{alg:ako-full:line:range-minimum}
\EndForEach
\ForEach{$s \in \set{-K, \dots, K}$} \label{alg:ako-full:line:iter-output}
    \State Let $\Delta_{v, s}$ be the result of the recovery algorithm (\cref{lem:precision-sampling}) applied to \label{alg:ako-full:line:recovery}
    \Statex[1] $\widetilde A_{0, s}, \dots, \widetilde A_{B-1, s}$ with  precisions $u_{v_0}, \dots, u_{v_{B-1}}$
\EndForEach
\State\Return $\min(\Delta_{v, s}, K)$ for all $s \in \set{-K, \dots, K}$ \label{alg:ako-full:line:return}
\end{algorithmic}
\end{algorithm}

In this section we give a formal proof of the correctness and running time analysis of (our reinterpretation of) the algorithm by Andoni, Krauthgamer and Onak~\cite{AndoniKO10}. We give the full algorithm in \cref{alg:ako-full}.

\paragraph{Setting the Parameters}
Throughout this section we assume that $T$ is a balanced $B$-ary partition tree with $n$ leaves. This means that the depth of $T$ is bounded by $\ceil{\log_B(n)}$ and each leaf $v$ is labeled with a singleton $I_v$. For each node $v$ in $T$, we define the following parameters:
\begin{itemize}
\item \emph{Rate $r_v$: If $v$ is the root, then we set $r_v = 1000/K$. Otherwise, if $v$ is a child of $w$ then sample~$u_v \sim \mathcal D((2 \log n)^{-1}, 0.01 \cdot (Kn)^{-1})$ independently and set $r_v = r_w / u_v$.}\\This assignment matches the values in \cref{alg:ako-full}.
\item \emph{Multiplicative accuracy $\alpha_v = 2 \cdot (1 - (2 \log n)^{-1})^d$ (where $d$ is the depth of $v$).}\\Note that $\alpha_v \geq 1$ since $d \leq \log n$.
\end{itemize}
Recall that the task of solving a node $v$ is computing a list $\Delta_{v, -K}, \dots, \Delta_{v, K}$ in the sense of \cref{def:problem}. We will in fact compute values with the stronger guarantee that $\Delta_{v, s}$ is a $(\alpha_v, 1/r_v)$-approximation of $\TD_{v, s}^{\leq K}(X, Y)$. Therefore, for the root node we compute an approximation of $\TD^{\leq K}(X, Y)$ with multiplicative error $2$ and additive error $0.001K$.

\begin{lemma}[Correctness of \cref{alg:ako-full}] \label{lem:ako-correctness}
Let $X, Y$ be strings. Given any node $v$ in the computation tree $v$, \cref{alg:ako-full} correctly solves $v$, with constant probability $0.9$.
\end{lemma}
\begin{proof}
We prove the lemma by induction over the depth of the computation tree, i.e., we assume that the algorithm correctly solves all recursive calls to the children of $v$, in the sense of the previous definition.

For the base case, we consider the leaf nodes handled in \crefrange{alg:ako-full:line:test-1-condition}{alg:ako-full:line:test-rate}. The first test in \crefrange{alg:ako-full:line:test-1-condition}{alg:ako-full:line:test-1} is correct since it computes the tree distance exactly, while the test in \crefrange{alg:ako-full:line:test-rate-condition}{alg:ako-full:line:test-rate} is correct since $|X_v| \leq 1/r_v$ implies that $0$ is an additive $1/r_v$-approximation of~$\TD^{\leq K}_{v, s}(X, Y) \leq |X_v|$.

For the inductive step, assume that all values $\Delta_{v_i, s'}$ recursively computed in \cref{alg:ako-full:line:recursion} are correct. (We will bound the error probability later.) The algorithm approximately evaluates the following expression (which is immediate from \cref{def:capped-td}):
\begin{equation*}
    \TD^{\leq K}_{v, s}(X, Y) = \min\Bigg(\sum_{i \in \rangezero B} A_{i, s},\,K\Bigg), \quad A_{i, s} = \min_{-K \leq s' \leq K} \TD^{\leq K}_{v_i, s'}(X, Y) + 2 \cdot |s - s'|.
\end{equation*}
Indeed, for any shift $s$, the value
\begin{equation*}
    \widetilde A_{i, s} = \min_{-K \leq s' \leq K} \Delta_{v_i, s'} + 2 \cdot |s - s'|,
\end{equation*}
as computed by \cref{alg:ako-full} in \cref{alg:ako-full:line:range-minimum}, constitutes a $(\alpha_{v_i}, 1/r_{v_i})$-approximation of $A_{i, s}$. Recall that the rates $r_{v_i}$ are set to $r_{v_i} = r_v / u_{v_i}$ for all children. Here, $u_{v_i}$ is an independent sample from $\mathcal D(\varepsilon = (\log n)^{-1}, \delta = 0.01 \cdot (Kn)^{-1})$. Hence, in \cref{alg:ako-full:line:recovery} we may apply the Precision Sampling Lemma to recover a $((1 + \varepsilon) \cdot \alpha_{v_i}, 1/r_v)$-approximation of $\sum_i A_{i, s}$. (Again, ignore the error probability for now.) By a simple calculation we obtain that $(1 + \varepsilon) \cdot \alpha_{v_i} = 2 \cdot (1 + (2 \log n)^{-1}) \cdot (1 - (2 \log n)^{-1})^d \leq \alpha_v$, and therefore the approximation quality of $\Delta_{v, s}$ is sufficient to satisfy the induction hypothesis.

It remains to bound the error probability. The only source of randomness in the algorithm is the precision sampling and the recovery using \cref{lem:precision-sampling}. The algorithm succeeds if the recovery algorithm in \cref{alg:ako-full:line:recovery} succeeds for all nodes $v$ and all values~\makebox{$s \in \set{-K, \dots, K}$}. Taking a union bound over these $2Kn$ many error events, each happening at most with probability $\delta = 0.01 \cdot (Kn)^{-1}$, we can bound the total error probability by $0.05$.
\end{proof}

\begin{lemma}[Running Time of \cref{alg:ako-full}] \label{lem:ako-time}
\cref{alg:ako-full} runs in time $nB \cdot (\log n)^{\Order(\log_B(n))}$, with constant probability $0.9$.
\end{lemma}
\begin{proof}
We first bound the running time of a single execution of \cref{alg:ako-full} (i.e., ignoring the cost of the recursive calls) by $\Order(K B \polylog(n))$:
\begin{itemize}
\item \crefrange{alg:ako-full:line:test-1-condition}{alg:ako-full:line:test-rate} and \cref{alg:ako-full:line:return}: These steps take time $\Order(K)$ to produce the output and are thus negligible. (Recall that the edit distance computation in \cref{alg:ako-full:line:test-1} boils down to comparing single characters.)
\item \crefrange{alg:ako-full:line:iter-children}{alg:ako-full:line:range-minimum}: The loop runs for $B$ iterations. In each iteration, we sample a precision in \cref{alg:ako-full:line:precision} in negligible time, perform a recursive computation which we ignore here, and apply \cref{lem:range-minimum} in time $\Order(K)$. The total time is $\Order(KB)$ as claimed.
\item \crefrange{alg:ako-full:line:iter-output}{alg:ako-full:line:recovery}: The loop runs for $\Order(K)$ iterations and in each iteration we apply the recovery algorithm from \cref{lem:precision-sampling} with parameters $\varepsilon = \Omega(\log n)$ and $\delta \geq n^{-\Order(1)}$. Each execution takes time $\Order(B \varepsilon^{-2} \log(\delta^{-1})) = \Order(B \polylog(n))$, as claimed.
\end{itemize}

Next, we argue that the number of active nodes in the computation tree is bounded by~$n / K \cdot (\log n)^{\Order(\log_B(n))}$ with good probability. It then follows that the total running time is bounded as claimed. To see this, we first analyze the rate $r_v$ at any node $v$. More specifically we apply \cref{lem:precision-sampling} with $N = 100n$ to obtain that: For any sample $u_v$ from $\mathcal D(\varepsilon=(2\log n)^{-1}, \delta=0.01\cdot (Kn)^{-1})$ there exists an event $E_v$ such that $\Pr(E_v) \geq 1 - 0.01 / n$ and
\begin{equation*}
    \Ex(1 / u_v \mid E_v) \leq \widetilde{\Order}(\varepsilon^{-2} \log(\delta^{-1}) \log n) = \polylog(n).
\end{equation*}
Using a union bound over all nodes $v$ in the tree, we can assume that all events $E_v$ simultaneously happen with probability at least $0.99$. Hence, from now on we condition on~$E = \bigwedge_v E_v$. Recall that $r_v = 1000 \cdot (K \cdot u_{v_1} \dots u_{v_d})^{-1}$ where $v_0, v_1, \dots, v_d = v$ is the root-to-node path leading to $v$. Therefore, and using that the $u_v$'s are sampled independently:
\begin{equation*}
    \Ex(r_v \mid E) = \frac{1000}K \cdot \prod_{i=1}^d \Ex(1 / u_{v_i} \mid E_{v_i}) \leq \frac1K \cdot (\log n)^{\Order(d)}.
\end{equation*}

We will now use that a node $v$ is active only if $|X_v| > 1/r_v$ (this is enforced by \crefrange{alg:ako-full:line:test-rate-condition}{alg:ako-full:line:test-rate}). Using Markov's inequality we obtain that
\begin{equation*}
    \Pr(\text{$v$ is active} \mid E) = \Pr(r_v > 1 / |X_v| \mid E) \leq |X_v| \cdot \Ex(r_v \mid E) \leq \frac{|X_v|}{K} \cdot (\log n)^{\Order(d)}.
\end{equation*}
Therefore, the number of active nodes at depth $d$ is $\sum_v |X_v| / K \cdot (\log n)^{\Order(d)} = n / K \cdot (\log n)^{\Order(d)}$ (where the sum is over all nodes $v$ at depth $d$, hence $\sum_v |X_v| = |X| = n$). We apply this bound at the deepest level $d = \log_B(n)$, and obtain by another application of Markov's inequality that the total number of active nodes is bounded by $100 n / K \cdot (\log n)^{\Order(\log_B(n))}$ with probability at least $0.99$. (The total success probability is $0.98$.)
\end{proof}
\section{2-Approximating Edit Distance for Many Shifts} \label{sec:ed-shifts}
In this section we show how to modify the Landau-Vishkin algorithm~\cite{LandauV88} to give a constant-factor approximation of the edit distance of a string $X$ and several consecutive shifts of another string $Y$. Recall that we use this routine at the leaves of our algorithm (\cref{alg:main:line:test-short} of \cref{alg:main}).

\thmedshifts*
\begin{proof}
\newcommand\D[2]{D[\,#1,#2\,]}
Let $m = |X|$ and $n = |Y|$. We will write $Y_s = Y\range{K+s}{m+K+s}$ for short. Consider the following table $\D ij$ defined as 
\begin{equation}\label{eqn:dp}
  \D ij = \min_{0 \leq s' \leq 2K} \ED(X\range{0}{i}, Y\range{s'}{j}),
\end{equation}
where $i \in \rangezero{m}$ and $j \in \rangezero{n}$.

First, we will explain how to compute the values $\D{m}{m + K + s}$ in the desired time, and then we show that these constitute a 2-approximation of $\ED^{\leq K}(X, Y_s)$.

\begin{itemize}[itemsep=\smallskipamount]
\itemdesc{Computing the Values $\D{m}{m + K + s}$:} It is not hard to see that \Cref{eqn:dp} follows the same recursive formulation as the dynamic program for edit distance, but with a twist in the base case. That is, we initialize $\D{0}{j} = 0$ for all $j \in \set{0, \dots, 2K}$ (which accounts for the possible starting shifts of $Y$) and the rest of the entries follow the recurrence
\begin{equation*}
  \D ij = \min\set{1 + \D{i-1}{j}, 1 + \D{i}{j-1}, c_{i,j} + \D{i-1}{j-1}},
\end{equation*}
where $c_{i,j}$ is the cost of matching $X \access i$ with $Y \access j$ (that is, $c_{i,j} = 1$ if $X\access{i} \neq Y\access{j}$, and otherwise $c_{i,j} = 0$) and where we set the out-of-bounds entries $\D{-1}{j}, \D{i}{-1}$ to $\infty$.

Since we want to approximate capped edit distances, we are only interested in computing values $\leq K$ in this table. This can be done using the classic Landau-Vishkin algorithm~\cite{LandauV88} in time $O(|X| + K^2)$. We now briefly sketch how this algorithm works. The idea is to iterate over the edit distance values $d = 0,\dots,K$. For each value $d$ and for each upper-left to bottom-right diagonal in the dynamic programming table, we maintain the furthest position to the right along each diagonal which contains the value $d$. We call this set of positions the \emph{frontier}. Since moving away from a diagonal incurs cost~1 and since we initialize $K$ consecutive diagonals with $d = 0$ in the base case, the frontier consists of $O(K)$ values at any iteration of the algorithm (i.e.\ we only need to keep track of $O(K)$ diagonals). Given the frontier for some edit distance value $d$, we can obtain the frontier for $d + 1$ by performing a longest common extension query for each diagonal and updating the corresponding position along each diagonal. Longest common extension queries can be answered using suffix trees in $O(1)$-time after preprocessing $X$ and $Y$ in linear time~\cite{LandauV88,Farach-ColtonFM00,KarkkainenSB06}, and updating the position along each diagonal also takes constant time. Thus, each of the $K$ iterations takes time $O(K)$, yielding the desired running time of $O(|X| + K^2)$.

\itemdesc{Approximation Guarantee:} We claim that for every $s \in \set{-K, \dots, K}$, the following bounds hold:
\begin{equation*}
  \D{m}{m + K + s} \leq \ED(X, Y_s) \leq 2 \cdot \D{m}{m + K + s}.
\end{equation*}
This claim implies the lemma statement. Indeed, the previous step allows us to compute the entries $\D{m}{m+K+s}$ which are at most $K$ and identify those which have value larger than $K$. By capping the latter entries at $K$, the claim implies that we get a 2\=/approximation for the capped distances $\ED^{\leq K}(X, Y_s)$ for every $s \in \set{-K, \dots, K}$, as desired.

We now prove the claimed bounds. The lower bound $\D{m}{m + K + s} \leq \ED(X, Y_s)$ holds by the definition of $\D{m}{m + K + s}$. For the upper bound, let $s'$ be the shift which minimizes the expression for $\D{m}{m+K+s}$ in \Cref{eqn:dp}, i.e., let $s'$ be such that $\D{m}{m+K+s} = \ED(X, Y\range{K+s'}{m+K+s})$. We will use that the edit distance of any two strings $A, B$ is at least their length difference $\big|\, |A| - |B| \,\big|$. In particular, we have~$\D{m}{m+K+s} \geq |s - s'|$. 

By the triangle inequality, we have that 
\begin{align*}
  \ED(X, Y_s) 
    & \leq \ED(X, Y\range{K+s'}{m+K+s}) + \ED(Y\range{K+s'}{m+K+s}, Y_s) \\
    & \leq \D{m}{m+K+s} + |s - s'| \leq 2 \cdot \D{m}{m+K+s},
\end{align*}
where the second inequality follows because we can transform $Y\range{K+s'}{m+K+s}$ into $Y_s = Y\range{K+s}{m+K+s}$ by deleting or inserting $|s - s'|$ characters. This proves the claim and thus finishes the proof of~\cref{thm:ed-shifts}. \qedhere




\end{itemize}

\end{proof}

\bibliography{refs}

\newcommand{\etalchar}[1]{$^{#1}$}
\begin{thebibliography}{AHVWW16}

\bibitem[ABVW15]{AbboudBW15}
Amir Abboud, Arturs Backurs, and Virginia Vassilevska~Williams.
\newblock Tight hardness results for {LCS} and other sequence similarity
  measures.
\newblock In {\em Proceedings of the 56th {IEEE} Annual Symposium on
  Foundations of Computer Science}, FOCS '15, pages 59--78. {IEEE} Computer
  Society, 2015.

\bibitem[AHVWW16]{AbboudHWW16}
Amir Abboud, Thomas~Dueholm Hansen, Virginia Vassilevska~Williams, and Ryan
  Williams.
\newblock Simulating branching programs with edit distance and friends: or: a
  polylog shaved is a lower bound made.
\newblock In {\em Proceedings of the 48th {ACM} Symposium on Theory of
  Computing}, STOC '16, pages 375--388. {ACM}, 2016.

\bibitem[AKO10]{AndoniKO10}
Alexandr Andoni, Robert Krauthgamer, and Krzysztof Onak.
\newblock Polylogarithmic approximation for edit distance and the asymmetric
  query complexity.
\newblock In {\em Proceedings of the 51st {IEEE} Annual Symposium on
  Foundations of Computer Science}, FOCS '10, pages 377--386. {IEEE} Computer
  Society, 2010.

\bibitem[AKO11]{AndoniKO11}
Alexandr Andoni, Robert Krauthgamer, and Krzysztof Onak.
\newblock Streaming algorithms via precision sampling.
\newblock In {\em Proceedings of the 52nd {IEEE} Annual Symposium on
  Foundations of Computer Science}, FOCS '11, pages 363--372. {IEEE} Computer
  Society, 2011.

\bibitem[AN20]{AndoniN20}
Alexandr Andoni and Negev~Shekel Nosatzki.
\newblock Edit distance in near-linear time: it's a constant factor.
\newblock In {\em Proceedings of the 61st {IEEE} Annual Symposium on
  Foundations of Computer Science}, FOCS '20, pages 990--1001. {IEEE} Computer
  Society, 2020.

\bibitem[And17]{Andoni17}
Alexandr Andoni.
\newblock High frequency moments via max-stability.
\newblock In {\em 2017 {IEEE} International Conference on Acoustics, Speech and
  Signal Processing, {ICASSP} 2017, New Orleans, LA, USA, March 5-9, 2017},
  pages 6364--6368. {IEEE}, 2017.

\bibitem[AO12]{AndoniO12}
Alexandr Andoni and Krzysztof Onak.
\newblock Approximating edit distance in near-linear time.
\newblock {\em {SIAM} J. Comput.}, 41(6):1635--1648, 2012.

\bibitem[BEK{\etalchar{+}}03]{BatuEKMRRS03}
Tugkan Batu, Funda Erg{\"{u}}n, Joe Kilian, Avner Magen, Sofya Raskhodnikova,
  Ronitt Rubinfeld, and Rahul Sami.
\newblock A sublinear algorithm for weakly approximating edit distance.
\newblock In {\em Proceedings of the 35th {ACM} Symposium on Theory of
  Computing}, {STOC '03}, pages 316--324. {ACM}, 2003.

\bibitem[BES06]{BatuES06}
Tugkan Batu, Funda Ergun, and Cenk Sahinalp.
\newblock Oblivious string embeddings and edit distance approximations.
\newblock In {\em Proceedings of the 17th {ACM-SIAM} Symposium on Discrete
  Algorithms}, SODA '06, pages 792--801. {SIAM}, 2006.

\bibitem[BI15]{BackursI15}
Arturs Backurs and Piotr Indyk.
\newblock Edit distance cannot be computed in strongly subquadratic time
  (unless {SETH} is false).
\newblock In {\em Proceedings of the 47th {ACM} Symposium on Theory of
  Computing}, STOC '15, pages 51--58. {ACM}, 2015.

\bibitem[BK15]{BringmannK15}
Karl Bringmann and Marvin K{\"{u}}nnemann.
\newblock Quadratic conditional lower bounds for string problems and dynamic
  time warping.
\newblock In {\em Proceedings of the 56th {IEEE} Annual Symposium on
  Foundations of Computer Science}, FOCS '15, pages 79--97. {IEEE} Computer
  Society, 2015.

\bibitem[BR20]{BrakensiekR20}
Joshua Brakensiek and Aviad Rubinstein.
\newblock Constant-factor approximation of near-linear edit distance in
  near-linear time.
\newblock In {\em Proceedings of the 52nd {ACM} Symposium on Theory of
  Computing}, STOC '20, pages 685--698. {ACM}, 2020.

\bibitem[BYJKK04]{BarYossefJKK04}
Ziv Bar-Yossef, S.~Thathachar Jayram, Robert Krauthgamer, and Ravi Kumar.
\newblock Approximating edit distance efficiently.
\newblock In {\em Proceedings of the 45th {IEEE} Annual Symposium on
  Foundations of Computer Science}, FOCS '04, pages 550--559. {IEEE} Computer
  Society, 2004.

\bibitem[CDG{\etalchar{+}}20]{ChakrabortyDGKS20}
Diptarka Chakraborty, Debarati Das, Elazar Goldenberg, Michal Kouck\'{y}, and
  Michael Saks.
\newblock Approximating edit distance within constant factor in truly
  sub-quadratic time.
\newblock {\em J. ACM}, 67(6), 2020.

\bibitem[FFM00]{Farach-ColtonFM00}
Martin Farach{-}Colton, Paolo Ferragina, and S.~Muthukrishnan.
\newblock On the sorting-complexity of suffix tree construction.
\newblock {\em J. {ACM}}, 47(6):987--1011, 2000.

\bibitem[GKKS21]{GoldenbergKKS21}
Elazar Goldenberg, Tomasz Kociumaka, Robert Krauthgamer, and Barna Saha.
\newblock Gap edit distance via non-adaptive queries: Simple and optimal.
\newblock {\em CoRR}, abs/2111.12706, 2021.

\bibitem[GKS19]{GoldenbergKS19}
Elazar Goldenberg, Robert Krauthgamer, and Barna Saha.
\newblock Sublinear algorithms for gap edit distance.
\newblock In {\em Proceedings of the 61st {IEEE} Annual Symposium on
  Foundations of Computer Science}, FOCS '20, pages 1101--1120. {IEEE} Computer
  Society, 2019.

\bibitem[Gus97]{Gusfield97}
Dan Gusfield.
\newblock {\em Algorithms on Strings, Trees, and Sequences: Computer Science
  and Computational Biology}.
\newblock Cambridge University Press, 1997.

\bibitem[KMP77]{KnuthMP77}
Donald~E. Knuth, James~H. Morris, and Vaughan~R. Pratt.
\newblock Fast pattern matching in strings.
\newblock {\em {SIAM} J. Comput.}, 6(2):323--350, 1977.

\bibitem[KS20a]{KociumakaS20}
Tomasz Kociumaka and Barna Saha.
\newblock Sublinear-time algorithms for computing {\&} embedding gap edit
  distance.
\newblock In {\em Proceedings of the 61st {IEEE} Annual Symposium on
  Foundations of Computer Science}, FOCS '20, pages 1168--1179. {IEEE} Computer
  Society, 2020.

\bibitem[KS20b]{KouckyS20}
Michal Kouck{\'{y}} and Michael~E. Saks.
\newblock Constant factor approximations to edit distance on far input pairs in
  nearly linear time.
\newblock In {\em Proceedings of the 52nd {ACM} Symposium on Theory of
  Computing}, STOC '20, pages 699--712. {ACM}, 2020.

\bibitem[KSB06]{KarkkainenSB06}
Juha K{\"{a}}rkk{\"{a}}inen, Peter Sanders, and Stefan Burkhardt.
\newblock Linear work suffix array construction.
\newblock {\em J. {ACM}}, 53(6):918--936, 2006.

\bibitem[Lev66]{Levenshtein66}
Vladimir~Iosifovich Levenshtein.
\newblock Binary codes capable of correcting deletions, insertions and
  reversals.
\newblock {\em Soviet Physics Doklady}, 10(8):707--710, 1966.

\bibitem[LMS98]{LandauMS98}
Gad~M. Landau, Eugene~W. Myers, and Jeanette~P. Schmidt.
\newblock Incremental string comparison.
\newblock {\em {SIAM} J. Comput.}, 27(2):557--582, 1998.

\bibitem[LV88]{LandauV88}
Gad~M. Landau and Uzi Vishkin.
\newblock Fast string matching with $k$ differences.
\newblock {\em J. Comput. Syst. Sci.}, 37(1):63--78, 1988.

\bibitem[OR07]{OstrovskyR07}
Rafail Ostrovsky and Yuval Rabani.
\newblock Low distortion embeddings for edit distance.
\newblock {\em J. {ACM}}, 54(5):23, 2007.

\bibitem[Vin68]{Vintsyuk68}
T.~K. Vintsyuk.
\newblock Speech discrimination by dynamic programming.
\newblock {\em Cybernetics}, 4(1):52--57, 1968.
\newblock Russian Kibernetika 4(1):81-88 (1968).

\bibitem[WF74]{WagnerF74}
Robert~A. Wagner and Michael~J. Fischer.
\newblock The string-to-string correction problem.
\newblock {\em J. {ACM}}, 21(1):168--173, 1974.

\end{thebibliography}

\end{document}